\documentclass[envcountsect]{llncs}

\usepackage{times}
\usepackage{graphicx}
\usepackage{amsmath}
\usepackage{amssymb}
\usepackage{thmtools}
\usepackage[noend]{algpseudocode}
\usepackage{algorithm}
\usepackage{enumerate}
\usepackage{tikz}
\usetikzlibrary{arrows,positioning,automata}
\usepackage{tabulary}
\usepackage{color}
\usepackage{dsfont}
\usepackage{multirow}
\usepackage{hyperref}

\spnewtheorem*{runningexample}{Running Example}{\itshape}{\rmfamily}

\newcommand{\st}{\:|\:}
\newcommand{\N}{\mathbb{N}}
\newcommand{\Z}{\mathbb{Z}}
\newcommand{\Q}{\mathbb{Q}}
\newcommand{\R}{\mathbb{R}}
\newcommand{\ind}[1]{\mathds{1}_{#1}}

\newcommand{\sharpP}{\ensuremath{\mathsf{\# P}}}
\newcommand{\PSPACE}{\ensuremath{\mathsf{PSPACE}}}

\newcommand{\hard}{\mathcal{H}}
\newcommand{\soft}{\mathcal{S}}
\newcommand{\improvs}{I}
\newcommand{\valids}{A}

\newcommand{\eopt}{\epsilon_\mathrm{opt}}

\newcommand{\rcic}[2]{\textup{RCI}\,\textup{(}#1\textup{,}\,#2\textup{)}}

\newcommand{\QBF}{QBF}

\newcommand{\DFA}{\textup{DFA}}
\newcommand{\NFA}{\textup{NFA}}
\newcommand{\CFG}{\textup{CFG}}

\newcommand{\LTL}{\textup{LTL}}
\newcommand{\LDL}{\textup{LDL}}

\newcommand{\alphabet}{\mathrm{\Sigma}}
\newcommand{\histories}{\alphabet^{\le n}}
\newcommand{\possible}{\mathrm{\Pi}}
\newcommand{\width}[1]{W(#1)}
\newcommand{\widthrel}[2]{W(#1|#2)}
\newcommand{\emptyword}{\lambda}
\newcommand{\pweight}[2]{C(#1, #2)}

\tikzstyle{adversary}=[rectangle,draw=black,minimum size=1cm]
\tikzstyle{soft}=[fill=black!20]

\allowdisplaybreaks

\begin{document}

\title{Reactive Control Improvisation}

\author{Daniel J. Fremont
%\orcidID{0000-0002-9992-9965}
 \and Sanjit A. Seshia
%\orcidID{0000-0001-6190-8707}
 }
\institute{University of California, Berkeley, USA\\
\email{\{dfremont, sseshia\}@berkeley.edu}}

\maketitle

\begin{abstract}
Reactive synthesis is a paradigm for automatically building correct-by-construction systems that interact with an unknown or adversarial environment.
We study how to do reactive synthesis when part of the specification of the system is that its behavior should be \emph{random}.
Randomness can be useful, for example, in a network protocol fuzz tester whose output should be varied, 
or a planner for a surveillance robot whose route should be unpredictable.
However, existing reactive synthesis techniques do not provide a way to ensure random behavior while maintaining functional correctness.
Towards this end, we generalize the recently-proposed framework of \emph{control improvisation} (CI) to add reactivity.
The resulting framework of {\em reactive control improvisation} 
provides a natural way to integrate a randomness requirement with the usual functional specifications of reactive synthesis
over a finite window.
We theoretically characterize when such problems are realizable, and give a general method for solving them.
For specifications given by reachability or safety games or by deterministic finite automata, our method yields a polynomial-time synthesis algorithm.
For various other types of specifications including temporal logic formulas, we obtain a polynomial-space algorithm and prove matching $\PSPACE$-hardness results.
We show that all of these randomized variants of reactive synthesis are no harder in a complexity-theoretic sense than their non-randomized counterparts.
\end{abstract}

%=================================================================
\section{Introduction}
\label{sec:intro}

Many interesting programs, including protocol handlers, task planners, and concurrent software generally, are \emph{open} systems that interact over time with an external environment.
Synthesis of such \emph{reactive systems} requires finding an implementation that satisfies the desired specification no matter what the environment does.
This problem, \emph{reactive synthesis}, has a long history (see \cite{finkbeiner-survey} for a survey).
Reactive synthesis from temporal logic specifications \cite{Pnueli1989} has been particularly well-studied and is being increasingly used in applications such as hardware synthesis \cite{BLOEM20073} and robotic task planning \cite{kress2009temporal}.

In this paper, we investigate how to synthesize reactive systems with \emph{random behavior}: in fact, systems where \emph{being random in a
prescribed way is part of their specification}.
This is in contrast to prior work on stochastic games where randomness is used to model uncertain environments or randomized strategies are merely allowed, not required.
Solvers for stochastic games may incidentally produce randomized strategies to satisfy a functional specification (and some types of specification, e.g.~multi-objective queries \cite{multiobjective}, may only be realizable by randomized strategies), but do not provide a general way to \emph{enforce} randomness.
Unlike most specifications used in reactive synthesis, our randomness requirement is a property of a system's \emph{distribution} of behaviors, not of an individual behavior.
While probabilistic specification languages like PCTL \cite{pctl} can capture some such properties, the simple and natural randomness requirement we study here cannot be concisely expressed by existing languages (even those as powerful as SGL \cite{sgl}).
Thus, \emph{randomized reactive synthesis} in our sense requires significantly different methods than those previously studied.

However, we argue that this type of synthesis is quite useful, because introducing randomness into the behavior of a system can often be beneficial, enhancing \emph{variety}, \emph{robustness}, and \emph{unpredictability}.
Example applications include:
\begin{itemize}
\itemsep3pt
\item Synthesizing a black-box fuzz tester for a network service, we want a program that not only conforms to the protocol (perhaps only most of the time) but can generate many different sequences of packets: randomness ensures this.
\item Synthesizing a controller for a robot exploring an unknown environment, randomness provides a low-memory way to increase coverage of the space.
It can also help to reduce systematic bias in the exploration procedure.
\item Synthesizing a controller for a patrolling surveillance robot, introducing randomness in planning makes the robot's future location harder to predict.
\end{itemize}

Adding randomness to a system in an \emph{ad hoc} way could easily compromise its correctness.
This paper shows how a randomness requirement can be integrated \emph{into the synthesis process}, ensuring correctness as well as allowing trade-offs to be explored: how much randomness can be added while staying correct, or how strong can a specification be while admitting a desired amount of randomness?

To formalize randomized reactive synthesis we build on the idea of \emph{control improvisation}, introduced in~\cite{donze-tr13},
formalized in~\cite{fsttcs}, and further generalized in~\cite{jacm}.
Control improvisation (CI) is the problem of constructing an \emph{improviser}, a probabilistic algorithm which generates finite words subject to three constraints: a \emph{hard constraint} that must always be satisfied, a \emph{soft constraint} that need only be satisfied with some probability, and a \emph{randomness constraint} that no word be generated with probability higher than a given bound.
We define \emph{reactive control improvisation} (RCI), where the improviser generates a word incrementally, alternating adding symbols with an adversarial environment.
To perform synthesis in a finite window, we encode functional specifications and environment assumptions into the hard constraint, while the soft and randomness constraints allow us to tune how randomness is added to the system.
The improviser obtained by solving the RCI problem is then a solution to the original synthesis problem.

The difficulty of solving reactive CI problems depends on the type of specification.
We study several types commonly used in reactive synthesis, including reachability games (and variants, e.g. safety games) and formulas in the temporal logics $\LTL$ and $\LDL$ \cite{ltl,ldl}.
We also investigate the specification types studied in \cite{jacm}, showing how the complexity of the CI problem changes when adding reactivity.
For every type of specification we obtain a randomized synthesis algorithm whose complexity matches that of ordinary reactive synthesis (in a finite window).
This suggests that reactive control improvisation should be feasible in applications like robotic task planning where reactive synthesis tools have proved effective.

In summary, the main contributions of this paper are:
\begin{itemize}
\item The reactive control improvisation (RCI) problem definition (Sec.~\ref{section:prob-defn}); \smallskip
\item The notion of \emph{width}, a quantitative generalization of ``winning'' game positions that measures \emph{how many ways} a player can win from that position (Sec.~\ref{section:existence}); \smallskip
\item A characterization of when RCI problems are realizable in terms of width, and an explicit construction of an improviser (Sec.~\ref{section:existence}); \smallskip
\item A general method for constructing efficient improvisation schemes (Sec.~\ref{section:generic}); \smallskip
\item A polynomial-time improvisation scheme for reachability/safety games and deterministic finite automaton specifications (Sec.~\ref{section:reachability}); \smallskip
\item $\PSPACE$-hardness results for many other specification types including temporal logics, and matching polynomial-space improvisation schemes (Sec.~\ref{section:temporal}).
\end{itemize}
Finally, Sec.~\ref{section:conclusion} summarizes our results and gives directions for future work.

%==================================================================================
\section{Background} \label{section:background}

\subsection{Notation}

Given an alphabet $\alphabet$, we write $|w|$ for the length of a finite word $w \in \alphabet^*$, $\emptyword$ for the empty word, $\alphabet^n$ for the words of length $n$, and $\alphabet^{\le n}$ for $\cup_{0 \le i \le n} \alphabet^i$, the set of all words of length at most $n$.
We abbreviate deterministic/nondeterministic finite automaton by $\DFA$/$\NFA$, and context-free grammar by $\CFG$.
For an instance $\mathcal{X}$ of any such formalism, which we call a \emph{specification}, we write $L(\mathcal{X})$ for the language (subset of $\alphabet^*$) it defines (note the distinction between a language and a representation thereof).
We view formulas of Linear Temporal Logic (LTL) \cite{ltl} and Linear Dynamic Logic (LDL) \cite{ldl} as specifications using their natural semantics on finite words (see \cite{ldl}).

We use the standard complexity classes $\sharpP$ and $\PSPACE$, and the $\PSPACE$-complete problem $\QBF$ of determining the truth of a quantified Boolean formula.
For background on these classes and problems see for example \cite{arora-barak}.

Some specifications we use as examples are \emph{reachability games} \cite{Mazala2002}, where players' actions cause transitions in a state space and the goal is to reach a target state.
We group these games, \emph{safety games} where the goal is to \emph{avoid} a set of states, and \emph{reach-avoid} games combining reachability and safety goals \cite{reach-avoid}, together as \emph{reachability/safety games} (RSGs).
We draw reachability games as graphs in the usual way: squares are adversary-controlled states, and states with a double border are target states.

\subsection{Synthesis Games}

Reactive control improvisation will be formalized in terms of a 2-player game which is essentially the standard \emph{synthesis game} used in reactive synthesis \cite{finkbeiner-survey}.
However, our formulation is slightly different for compatibility with the definition of control improvisation, so we give a self-contained presentation here.

Fix a finite alphabet $\alphabet$.
The players of the game will alternate picking symbols from $\alphabet$, building up a word.
We can then specify the set of winning plays with a language over $\alphabet$.
To simplify our presentation we assume that players strictly alternate turns and that any symbol from $\alphabet$ is a legal move.
These assumptions can be relaxed in the usual way by modifying the winning set appropriately.

\textbf{Finite words:} While reactive synthesis is usually considered over infinite words, in this paper we focus on synthesis in a finite window, as it is unclear how best to generalize our randomness requirement to the infinite case.
This assumption is not too restrictive, as solutions of bounded length are adequate for many applications.
In fuzz testing, for example, we do not want to generate arbitrarily long files or sequences of packets.
In robotic planning, we often want a plan that accomplishes a task within a certain amount of time.
Furthermore, planning problems with liveness specifications can often be segmented into finite pieces: we do not need an infinite route for a patrolling robot, but can plan within a finite horizon and replan periodically.
Replanning may even be \emph{necessary} when environment assumptions become invalid.
At any rate, we will see that the bounded case of reactive control improvisation is already highly nontrivial.

As a final simplification, we require that all plays have length exactly $n \in \N$.
To allow a range $[m, n]$ we can simply add a new padding symbol to $\alphabet$ and extend all shorter words to length $n$, modifying the winning set appropriately.

\begin{definition}
A \emph{history} $h$ is an element of $\alphabet^{\le n}$, representing the moves of the game played so far.
We say the game has \emph{ended} after $h$ if $|h| = n$; otherwise it is \emph{our turn} after $h$ if $|h|$ is even, and \emph{the adversary's turn} if $|h|$ is odd.
\end{definition}

\begin{definition}
A \emph{strategy} is a function $\sigma : \histories \times \alphabet \rightarrow [0,1]$ such that for any history $h \in \histories$ with $|h| < n$, $\sigma(h, \cdot)$ is a probability distribution over $\alphabet$.
We write $x \leftarrow \sigma(h)$ to indicate that $x$ is a symbol randomly drawn from $\sigma(h, \cdot)$.
\end{definition}

Since strategies are randomized, fixing strategies for both players does not uniquely determine a play of the game, but defines a \emph{distribution} over plays:
\begin{definition}
Given a pair of strategies $(\sigma, \tau)$, we can generate a random \emph{play} $\pi \in \alphabet^n$ as follows.
Pick $\pi_0 \leftarrow \sigma(\emptyword)$, then for $i$ from $1$ to $n-1$ pick $\pi_i \leftarrow \tau(\pi_0 \dots \pi_{i-1})$ if $i$ is odd and $\pi_i \leftarrow \sigma(\pi_0 \dots \pi_{i-1})$ otherwise.
Finally, put $\pi = \pi_0 \dots \pi_{n-1}$.
We write $P_{\sigma,\tau}(\pi)$ for the probability of obtaining the play $\pi$.
This extends to a \emph{set} of plays $X \subseteq \alphabet^n$ in the natural way: $P_{\sigma,\tau}(X) = \sum_{\pi \in X} P_{\sigma,\tau}(\pi)$.
Finally, the set of \emph{possible} plays is $\possible_{\sigma,\tau} = \{ \pi \in \alphabet^n \st P_{\sigma,\tau}(\pi) > 0 \}$.

\end{definition}

The next definition is just the conditional probability of a play given a history, but works for histories with probability zero, simplifying our presentation.
\begin{definition}
For any history $h = h_0 \dots h_{k-1} \in \histories$ and word $\rho \in \alphabet^{n - k}$, we write $P_{\sigma,\tau}(\rho | h)$ for the probability that if we assign $\pi_i = h_i$ for $i < k$ and sample $\pi_k, \dots, \pi_{n-1}$ by the process above, then $\pi_k \dots \pi_{n-1} = \rho$.
\end{definition}

\section{Problem Definition} \label{section:prob-defn}

\subsection{Motivating Example}

Consider synthesizing a planner for a surveillance drone operating near another, potentially adversarial drone.
Discretizing the map into the 7x7 grid in Fig.~\ref{figure:drones} (ignoring the depicted trajectories for the moment), a route is a word over the four movement directions.
Our specification is to visit the 4 circled locations in 30 moves without colliding with the adversary, assuming it cannot move into the 5 highlighted central locations.

\begin{figure}[tb]
\begin{minipage}{0.5\textwidth}
\centering
\includegraphics[width=0.7\textwidth]{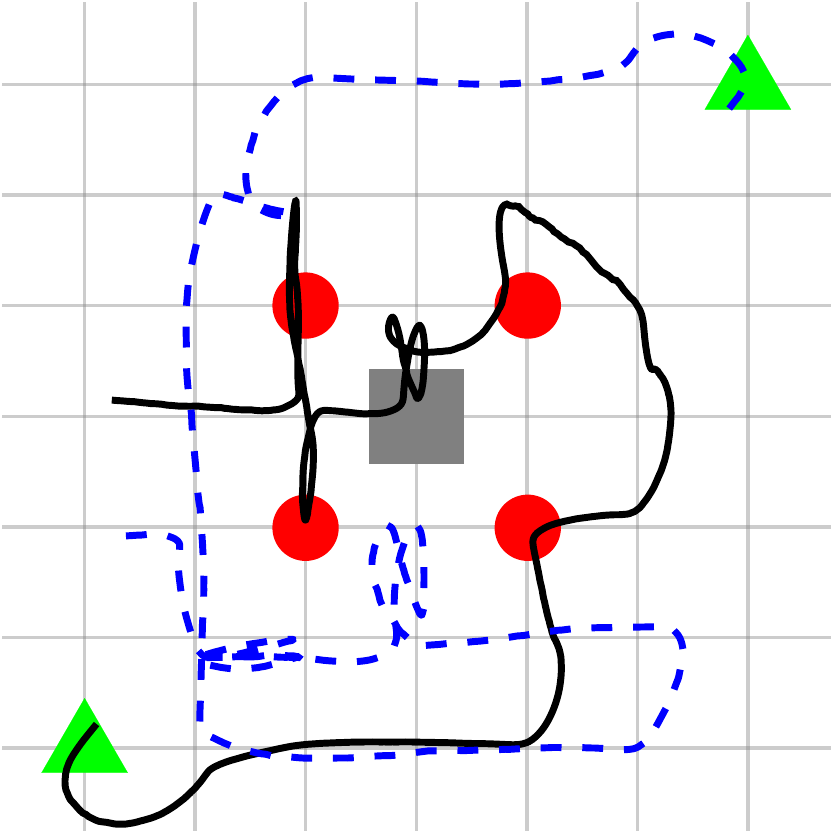}
\end{minipage}
\begin{minipage}{0.5\textwidth}
\centering
\includegraphics[width=0.7\textwidth]{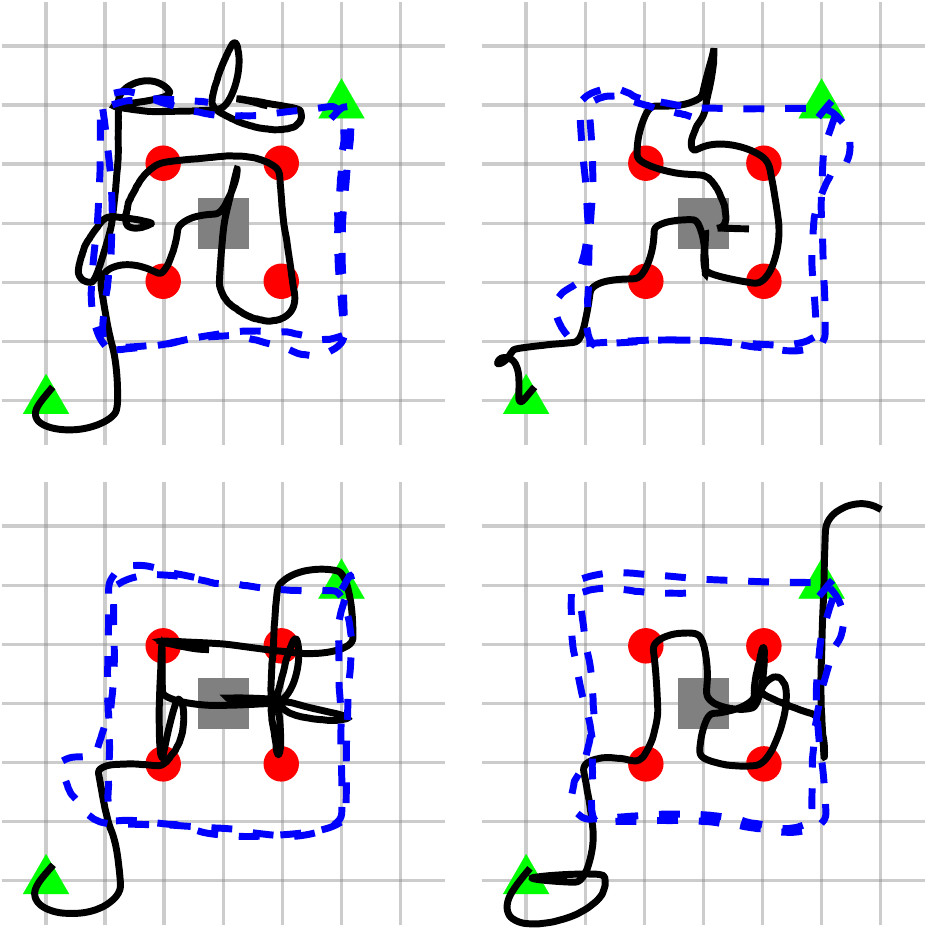}
\end{minipage}
\caption{Improvised trajectories for a patrolling drone (solid) avoiding an adversary (dashed).
The adversary may not move into the circles or the square.}
\label{figure:drones}
\end{figure}

Existing reactive synthesis tools can produce a strategy for the patroller ensuring that the specification is always satisfied.
However, the strategy may be deterministic, so that in response to a fixed adversary the patroller will always follow the same route.
Then it is easy for a third party to predict the route, which could be undesirable, and is in fact unnecessary if there are many other ways the drone can satisfy its specification.

Reactive control improvisation addresses this problem by adding a new type of specification to the \emph{hard constraint} above: a \emph{randomness requirement} stating that no behavior should be generated with probability greater than a threshold $\rho$.
If we set (say) $\rho = 1/5$, then any controller solving the synthesis problem must be able to satisfy the hard constraint in at least 5 different ways, never producing any given behavior more than 20\% of the time.
Our synthesis algorithm can in fact compute the smallest $\rho$ for which synthesis is possible, yielding a controller that is \emph{maximally-randomized} in that the system's behavior is as close to a uniform distribution as possible.

To allow finer tuning of how randomness is introduced into the controller, our definition also includes a \emph{soft constraint} which need only be satisfied with some probability $1 - \epsilon$.
This allows us to prefer certain safe behaviors over others.
In our drone example, we require that with probability at least $3/4$, we do not visit a circled location twice.

These hard, soft, and randomness constraints form an instance of our reactive control improvisation problem.
Encoding the hard and soft constraints as DFAs, our algorithm (Sec.~\ref{section:reachability}) produced a controller achieving the smallest realizable $\rho = 2.2 \times 10^{-12}$.
We tested the controller using the PX4 autopilot \cite{px4} to refine the generated routes into control actions for a drone simulated in Gazebo \cite{gazebo} (videos and code are available online \cite{videos}).
A selection of resulting trajectories are shown in Fig.~\ref{figure:drones} (the remainder in Appendix~\ref{section:drone-experiments}%
%of the full paper~\cite{full-version}
): starting from the triangles, the patroller's path is solid, the adversary's dashed.
The left run uses an adversary that moves towards the patroller when possible.
The right runs, with a simple adversary moving in a fixed loop, illustrate the randomness of the synthesized controller.

\subsection{Reactive Control Improvisation}

Our formal notion of randomized reactive synthesis in a finite window is a reactive extension of \emph{control improvisation} \cite{jacm,fsttcs}, which captures the three types of constraint (hard, soft, randomness) seen above.
We use the notation of~\cite{jacm} for the specifications and languages defining the hard and soft constraints:

\begin{definition}[\cite{jacm}]
Given \emph{hard} and \emph{soft} specifications $\hard$ and $\soft$ of languages over $\alphabet$, an \emph{improvisation} is a word $w \in L(\hard) \cap \alphabet^n$.
It is \emph{admissible} if $w \in L(\soft)$.
The set of all improvisations is denoted $\improvs$, and admissible improvisations $\valids$.
\end{definition}

\begin{runningexample}
We will use the following simple example throughout the paper: each player may increment ($+$), decrement ($-$), or leave unchanged ($=$) a counter which is initially zero.
The alphabet is $\alphabet = \{ +, -, = \}$, and we set $n = 4$.
The hard specification $\hard$ is the DFA in Fig.~\ref{figure:running-spec} requiring that the counter stay within $[-2, 2]$.
The soft specification $\soft$ is a similar DFA requiring that the counter end at a nonnegative value.

Then for example the word $+$$+$$=$$=$ is an admissible improvisation, satisfying both hard and soft constraints, and so is in $\valids$.
The word $+$$-$$=$$-$ on the other hand satisfies $\hard$ but not $\soft$, so it is in $\improvs$ but not $\valids$.
Finally, $+$$+$$+$$-$ does not satisfy $\hard$, so it is not an improvisation at all and is not in $\improvs$.
\end{runningexample}

\begin{figure}[tb]
\centering
\begin{tikzpicture}[initial text=, transform shape, scale=0.6, node distance=0.75cm]

 \node[state, soft, accepting, initial below] (s0) {\Large $+0$}; 
 \node[state, soft, accepting, right= of s0] (s1) {\Large $+1$}; 
 \node[state, soft, accepting, right= of s1] (s2) {\Large $+2$}; 
 \node[state, right= of s2] (s3) {\Large $+3$}; 
 \node[state, accepting, left= of s0] (s-1) {\Large $-1$}; 
 \node[state, accepting, left= of s-1] (s-2) {\Large $-2$}; 
 \node[state, left= of s-2] (s-3) {\Large $-3$}; 
 
 \path[->]
 (s-2) edge[bend left] node[above] {\Large $+$} (s-1)
 (s-1) edge[bend left] node[above] {\Large $+$} (s0)
 (s0) edge[bend left] node[above] {\Large $+$} (s1)
 (s1) edge[bend left] node[above] {\Large $+$}(s2)
 (s2) edge[bend left] node[above] {\Large $+$} (s3)
 (s2) edge[bend left] node[below] {\Large $-$} (s1)
 (s1) edge[bend left] node[below] {\Large $-$} (s0)
 (s0) edge[bend left] node[below] {\Large $-$} (s-1)
 (s-1) edge[bend left] node[below] {\Large $-$}(s-2)
 (s-2) edge[bend left] node[below] {\Large $-$} (s-3)
 (s-3) edge[loop left] node[left] {\Large $\alphabet$} (s-3)
 (s-2) edge[loop above] node[above] {\Large =} (s-2)
 (s-1) edge[loop above] node[above] {\Large =} (s-1)
 (s0) edge[loop above] node[above] {\Large =} (s0)
 (s1) edge[loop above] node[above] {\Large =} (s1)
 (s2) edge[loop above] node[above] {\Large =} (s2)
 (s3) edge[loop right] node[right] {\Large $\alphabet$} (s3);
 
 \end{tikzpicture}
\caption{The hard specification DFA $\hard$ in our running example. The soft specification $\soft$ is the same but with only the shaded states accepting.}
\label{figure:running-spec}
\end{figure}
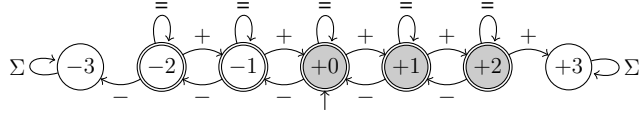

A reactive control improvisation problem is defined by $\hard$, $\soft$, and parameters $\epsilon$ and $\rho$.
A solution is then a strategy which ensures that the hard, soft, and randomness constraints hold against every adversary.
Formally, following \cite{jacm,fsttcs}:
\begin{definition}
Given an \emph{RCI instance} $\mathcal{C} = (\hard, \soft, n, \epsilon, \rho)$ with $\hard$, $\soft$, and $n$ as above and $\epsilon, \rho \in [0, 1] \cap \Q$, a strategy $\sigma$ is an \emph{improvising strategy} if it satisfies the following requirements for every adversary $\tau$:
\begin{description}
\item[Hard constraint:] $P_{\sigma,\tau} (\improvs) = 1$ \vspace{1pt}
\item[Soft constraint:] $P_{\sigma,\tau} (\valids) \ge 1 - \epsilon$ \vspace{1pt}
\item[Randomness:] $\forall \pi \in \improvs$, $P_{\sigma,\tau} (\pi) \le \rho$.
\end{description}
If there is an improvising strategy $\sigma$, we say that $\mathcal{C}$ is \emph{realizable}.
An \emph{improviser} for $\mathcal{C}$ is then an expected-finite time probabilistic algorithm implementing such a strategy $\sigma$, i.e. whose output distribution on input $h \in \histories$ is $\sigma(h, \cdot)$.
\end{definition}

\begin{definition}
Given an RCI instance $\mathcal{C} = (\hard, \soft, n, \epsilon, \rho)$, the \emph{reactive control improvisation} (RCI) problem is to decide whether $\mathcal{C}$ is realizable, and if so to generate an improviser for $\mathcal{C}$.
\end{definition}

\begin{runningexample}
Suppose we set $\epsilon = 1/2$ and $\rho = 1/2$.
Let $\sigma$ be the strategy which picks $+$ or $-$ with equal probability in the first move, and thenceforth picks the action which moves the counter closest to $\pm 1$ respectively.
This satisfies the hard constraint, since if the adversary ever moves the counter to $\pm 2$ we immediately move it back.
The strategy also satisfies the soft constraint, since with probability $1/2$ we set the counter to $+1$ on the first move, and if the adversary moves to $0$ we move back to $+1$ and remain nonnegative.
Finally, $\sigma$ also satisfies the randomness constraint, since each choice of first move happens with probability $1/2$ and so no play can be generated with higher probability.
So $\sigma$ is an improvising strategy and this RCI instance is realizable.
\end{runningexample}

We will study classes of RCI problems with different types of specifications:
\begin{definition}
If $\textsc{HSpec}$ and $\textsc{SSpec}$ are classes of specifications, then the class of RCI instances $\mathcal{C} = (\hard, \soft, n, \epsilon, \rho)$ where $\hard \in \textsc{HSpec}$ and $\soft \in \textsc{SSpec}$ is denoted $\rcic{\textsc{HSpec}}{\textsc{SSpec}}$.
We use the same notation for the decision problem associated with the class, i.e., given $\mathcal{C} \in \rcic{\textsc{HSpec}}{\textsc{SSpec}}$, decide whether $\mathcal{C}$ is realizable.
The \emph{size} $|\mathcal{C}|$ of an RCI instance is the total size of the bit representations of its parameters, with $n$ represented in unary and $\epsilon, \rho$ in binary.
\end{definition}

Finally, a \emph{synthesis algorithm} in our context takes a specification in the form of an RCI instance and produces an implementation in the form of an improviser.
This corresponds exactly to the notion of an improvisation scheme from \cite{jacm}:
\begin{definition}[\cite{jacm}] \label{defn:scheme}
A \emph{polynomial-time improvisation scheme} for a class $\mathcal{P}$ of RCI instances is an algorithm $S$ with the following properties: 
\begin{description}
\item[Correctness:] For any $\mathcal{C} \in \mathcal{P}$, if $\mathcal{C}$ is
realizable then $S(\mathcal{C})$ is an improviser for $\mathcal{C}$, and
otherwise $S(\mathcal{C}) = \bot$. \vspace{1pt}
\item[Scheme efficiency:] There is a polynomial $p : \R \rightarrow \R$ such that the runtime of $S$ on any $\mathcal{C} \in \mathcal{P}$ is at most $p(|\mathcal{C}|)$. \vspace{1pt}
\item[Improviser efficiency:] There is a polynomial $q : \R \rightarrow \R$ such that for every $\mathcal{C} \in \mathcal{P}$, if $G = S(\mathcal{C}) \ne \bot$ then $G$ has expected runtime at most $q(|\mathcal{C}|)$. 
\end{description}
\end{definition}
The first two requirements simply say that the scheme produces valid improvisers in polynomial time.
The third is necessary to ensure that the improvisers themselves are efficient: otherwise, the scheme might for example produce improvisers running in time exponential in the size of the specification.

A main goal of our paper is to determine for which types of specifications there exist polynomial-time improvisation schemes.
While we do find such algorithms for important classes of specifications, we will also see that determining the realizability of an RCI instance is often $\PSPACE$-hard.
Therefore we also consider \emph{polynomial-space improvisation schemes}, defined as above but replacing time with space.

\section{Existence of Improvisers} \label{section:existence}

\subsection{Width and Realizability}

The most basic question in reactive synthesis is whether a specification is realizable.
In \emph{randomized} reactive synthesis, the question is more delicate because the randomness requirement means that it is no longer enough to ensure some property regardless of what the adversary does: there must be \emph{many ways} to do so.
Specifically, there must be at least $1 / \rho$ improvisations if we are to generate each of them with probability at most $\rho$.
Furthermore, at least this many improvisations must be \emph{possible} given an unknown adversary: even if many exist, the adversary may be able to force us to use only a single one.
We introduce a new notion of the size of a set of plays that takes this into account.

\begin{definition}
The \emph{width} of $X \subseteq \alphabet^n$ is $\width{X} = \max_\sigma \min_\tau | X \cap \possible_{\sigma,\tau} |$.
\end{definition}

The width counts how many distinct plays can be generated regardless of what the adversary does.
Intuitively, a ``narrow'' game --- one whose set of winning plays has small width --- is one in which the adversary can force us to choose among only a few winning plays, while in a ``wide'' one we always have many safe choices available.
Note that \emph{which} particular plays can be generated depends on the adversary: the width only measures \emph{how many} can be generated.
For example, $\width{X} = 1$ means that a play in $X$ can always be generated, but possibly a different element of $X$ for different adversaries.

\begin{runningexample}
Figure~\ref{figure:running-game} shows the synthesis game for our running example: paths ending in circled or shaded states are plays in $\improvs$ or $\valids$ respectively (ignore the state labels for now).
At left, the bold arrows show the 4 plays in $\improvs$ possible against the adversary that moves away from 0, and down at 0.
This shows $\width{\improvs} \le 4$, and in fact 4 plays are possible against any adversary, so $\width{\improvs} = 4$.
Similarly, at right we see that $\width{\valids} = 1$.
\end{runningexample}

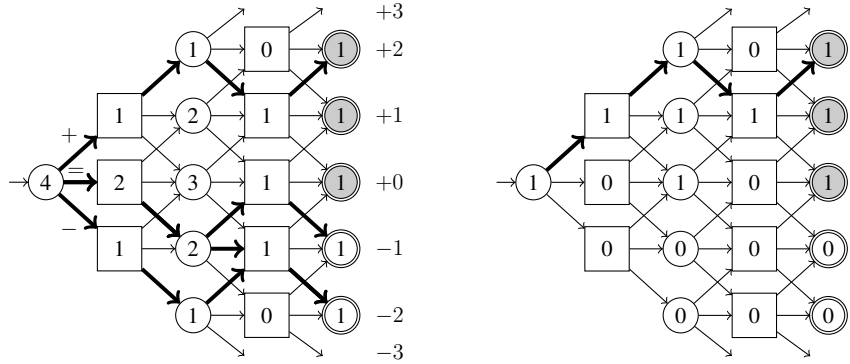
\begin{figure}[tb]
\begin{minipage}{0.5\textwidth}
\centering
\begin{tikzpicture}[initial text=, transform shape, scale=0.59, node distance=0.75cm]

 \node[state, initial] (s0) {\Large 4}; 
 \node[adversary, right= of s0, yshift=1.5cm] (s11) {\Large 1}; 
 \node[adversary, right= of s0] (s10) {\Large 2};
  \node[adversary, right= of s0, yshift=-1.5cm] (s1-1) {\Large 1};
 \node[state, right= of s11, yshift=1.5cm] (s22) {\Large 1};
 \node[state, right= of s11] (s21) {\Large 2};
 \node[state, right= of s10] (s20) {\Large 3};
 \node[state, right= of s1-1] (s2-1) {\Large 2};
 \node[state, right= of s1-1, yshift=-1.5cm] (s2-2) {\Large 1};
 \node[adversary, right= of s21, yshift=1.5cm] (s32) {\Large 0};
 \node[adversary, right= of s21] (s31) {\Large 1};
 \node[adversary, right= of s20] (s30) {\Large 1};
 \node[adversary, right= of s2-1] (s3-1) {\Large 1};
 \node[adversary, right= of s2-1, yshift=-1.5cm] (s3-2) {\Large 0};
 \node[state, accepting, soft, right= of s31, yshift=1.5cm] (s42) {\Large 1};
 \node[state, accepting, soft, right= of s31] (s41) {\Large 1};
 \node[state, accepting, soft, right= of s30] (s40) {\Large 1};
 \node[state, accepting, right= of s3-1] (s4-1) {\Large 1};
 \node[state, accepting, right= of s3-1, yshift=-1.5cm] (s4-2) {\Large 1};
 \node[right= of s22, yshift=1cm] (s33) {};
 \node[right= of s2-2, yshift=-1cm] (s3-3) {};
 \node[right= of s32, yshift=1cm] (s43) {};
 \node[right= of s3-2, yshift=-1cm] (s4-3) {};
 \node[right= of s42, xshift=-0.5cm] (l2) {\Large $+2$};
 \node[right= of s41, xshift=-0.5cm] (l1) {\Large $+1$};
 \node[right= of s40, xshift=-0.5cm] (l0) {\Large $+0$};
 \node[right= of s4-1, xshift=-0.5cm] (l-1) {\Large $-1$};
 \node[right= of s4-2, xshift=-0.5cm] (l-2) {\Large $-2$};
 \node[above= of l2, yshift=-0.5cm] (l3) {\Large $+3$};
 \node[below= of l-2, yshift=0.5cm] (l-3) {\Large $-3$};

 \path[->] 
 %(s0) edge (s11)
 %(s0) edge (s10)
 %(s0) edge (s1-1)
 %(s11) edge (s22)
 (s11) edge (s21)
 (s11) edge (s20)
 (s10) edge (s21)
 (s10) edge (s20)
 %(s10) edge (s2-1)
 (s1-1) edge (s20)
 (s1-1) edge (s2-1)
 %(s1-1) edge (s2-2)
 (s22) edge (s33)
 (s22) edge (s32)
 %(s22) edge (s31)
 (s21) edge (s32)
 (s21) edge (s31)
 (s21) edge (s30)
 (s20) edge (s31)
 (s20) edge (s30)
 (s20) edge (s3-1)
 %(s2-1) edge (s30)
 %(s2-1) edge (s3-1)
 (s2-1) edge (s3-2)
 %(s2-2) edge (s3-1)
 (s2-2) edge (s3-2)
 (s2-2) edge (s3-3)
 (s32) edge (s43)
 (s32) edge (s42)
 (s32) edge (s41)
 %(s31) edge (s42)
 (s31) edge (s41)
 (s31) edge (s40)
 (s30) edge (s41)
 (s30) edge (s40)
 %(s30) edge (s4-1)
 (s3-1) edge (s40)
 (s3-1) edge (s4-1)
 %(s3-1) edge (s4-2)
 (s3-2) edge (s4-1)
 (s3-2) edge (s4-2)
 (s3-2) edge (s4-3);

 \draw[->, line width=1.5pt]
 (s0) -- node[above, xshift=-0.2cm, yshift=0.1cm] {\Large $+$} (s11);
 \path[->, line width=1.5pt]
 (s11) edge (s22)
 (s22) edge (s31)
 (s31) edge (s42);
 \draw[->, line width=1.5pt]
 (s0) -- node[above, xshift=-0.1cm] {\Large $=$} (s10);
 \path[->, line width=1.5pt]
 (s10) edge (s2-1)
 (s2-1) edge (s3-1)
 (s3-1) edge (s4-2);
 \path[->, line width=1.5pt]
 (s0) edge (s10)
 (s10) edge (s2-1)
 (s2-1) edge (s30)
 (s30) edge (s4-1);
  \draw[->, line width=1.5pt]
 (s0) -- node[below, xshift=-0.2cm, yshift=-0.1cm] {\Large $-$} (s1-1);
  \path[->, line width=1.5pt]
 (s1-1) edge (s2-2)
 (s2-2) edge (s3-1)
 (s3-1) edge (s4-2);

\end{tikzpicture}
\end{minipage}
\begin{minipage}{0.5\textwidth}
\centering
\begin{tikzpicture}[initial text=, transform shape, scale=0.59, node distance=0.75cm]

 \node[state, initial] (s0) {\Large 1}; 
 \node[adversary, right= of s0, yshift=1.5cm] (s11) {\Large 1}; 
 \node[adversary, right= of s0] (s10) {\Large 0};
  \node[adversary, right= of s0, yshift=-1.5cm] (s1-1) {\Large 0};
 \node[state, right= of s11, yshift=1.5cm] (s22) {\Large 1};
 \node[state, right= of s11] (s21) {\Large 1};
 \node[state, right= of s10] (s20) {\Large 1};
 \node[state, right= of s1-1] (s2-1) {\Large 0};
 \node[state, right= of s1-1, yshift=-1.5cm] (s2-2) {\Large 0};
 \node[adversary, right= of s21, yshift=1.5cm] (s32) {\Large 0};
 \node[adversary, right= of s21] (s31) {\Large 1};
 \node[adversary, right= of s20] (s30) {\Large 0};
 \node[adversary, right= of s2-1] (s3-1) {\Large 0};
 \node[adversary, right= of s2-1, yshift=-1.5cm] (s3-2) {\Large 0};
 \node[state, accepting, soft, right= of s31, yshift=1.5cm] (s42) {\Large 1};
 \node[state, accepting, soft, right= of s31] (s41) {\Large 1};
 \node[state, accepting, soft, right= of s30] (s40) {\Large 1};
 \node[state, accepting, right= of s3-1] (s4-1) {\Large 0};
 \node[state, accepting, right= of s3-1, yshift=-1.5cm] (s4-2) {\Large 0};
 \node[right= of s22, yshift=1cm] (s33) {};
 \node[right= of s2-2, yshift=-1cm] (s3-3) {};
 \node[right= of s32, yshift=1cm] (s43) {};
 \node[right= of s3-2, yshift=-1cm] (s4-3) {};

 \path[->] 
 %(s0) edge (s11)
 (s0) edge (s10)
 (s0) edge (s1-1)
 (s11) edge (s22)
 (s11) edge (s21)
 (s11) edge (s20)
 (s10) edge (s21)
 (s10) edge (s20)
 (s10) edge (s2-1)
 (s1-1) edge (s20)
 (s1-1) edge (s2-1)
 (s1-1) edge (s2-2)
 (s22) edge (s33)
 (s22) edge (s32)
 %(s22) edge (s31)
 (s21) edge (s32)
 (s21) edge (s31)
 (s21) edge (s30)
 (s20) edge (s31)
 (s20) edge (s30)
 (s20) edge (s3-1)
 (s2-1) edge (s30)
 (s2-1) edge (s3-1)
 (s2-1) edge (s3-2)
 (s2-2) edge (s3-1)
 (s2-2) edge (s3-2)
 (s2-2) edge (s3-3)
 (s32) edge (s43)
 (s32) edge (s42)
 (s32) edge (s41)
 %(s31) edge (s42)
 (s31) edge (s41)
 (s31) edge (s40)
 (s30) edge (s41)
 (s30) edge (s40)
 (s30) edge (s4-1)
 (s3-1) edge (s40)
 (s3-1) edge (s4-1)
 (s3-1) edge (s4-2)
 (s3-2) edge (s4-1)
 (s3-2) edge (s4-2)
 (s3-2) edge (s4-3);

 \path[->, line width=1.5pt]
 (s0) edge (s11)
 (s11) edge (s22)
 (s22) edge (s31)
 (s31) edge (s42);

\end{tikzpicture}
\end{minipage}
\caption{Synthesis game for our running example. States are labeled with the widths of $\improvs$ (left) and $\valids$ (right) given a history ending at that state.}
\label{figure:running-game}
\end{figure}

It will be useful later to have a \emph{relative} version of width that counts how many plays are possible \emph{from a given position}:

\begin{definition}
Given a set of plays $X \subseteq \alphabet^n$ and a history $h \in \histories$, the \emph{width of} $X$ \emph{given} $h$ is $\widthrel{X}{h} = \max_\sigma \min_\tau | \{ \pi \st h \pi \in X \land P_{\sigma,\tau}(\pi | h) > 0 \} |$.
\end{definition}
This is a direct generalization of ``winning'' positions: if $X$ is the set of winning plays, then $\widthrel{X}{h}$ counts the number of ways to win from $h$.

We will often use the following basic properties of $\widthrel{X}{h}$ without comment %
(for this proof, and the details of later proof sketches, see
%(for lack of space this proof and the details of later proof sketches are deferred to
Appendix~\ref{section:detailed-proofs}%
% of the full paper~\cite{full-version}
).
Note that (3)--(5) provide a recursive way to compute widths that we will use later, and which is illustrated by the state labels in Fig.~\ref{figure:running-game}.
\begin{restatable}{lemma}{lemmaWidthrel} \label{lemma:widthrel}
For any set of plays $X \subseteq \alphabet^n$ and history $h \in \histories$:
\begin{enumerate}
 \item $0 \le \widthrel{X}{h} \le |\alphabet|^{n - |h|}$; \vspace{1pt}
 \item $\widthrel{X}{\emptyword} = \width{X}$; \vspace{1pt}
 \item if $|h| = n$, then $\widthrel{X}{h} = \ind{h \in X}$; \vspace{1pt}
 \item if it is our turn after $h$, then $\widthrel{X}{h} = \sum_{u \in \alphabet} \widthrel{X}{hu}$; \vspace{1pt}
 \item if it is the adversary's turn after $h$, then $\widthrel{X}{h} = \min_{u \in \alphabet} \widthrel{X}{hu}$.
\end{enumerate}
\end{restatable}

Now we can state the realizability conditions, which are simply that $\improvs$ and $\valids$ have sufficiently large width.
In fact, the conditions turn out to be exactly the same as those for non-reactive CI except that width takes the place of size \cite{fsttcs}.
\begin{theorem} \label{theorem:feasibility}
The following are equivalent:
\begin{enumerate}[(1)]
\item $\mathcal{C}$ is realizable.
\item $\width{\improvs} \ge 1/\rho$ and $\width{\valids} \ge (1-\epsilon) / \rho$.
\item There is an improviser for $\mathcal{C}$.
\end{enumerate}
\end{theorem}

\begin{runningexample}
We saw above that our example was realizable with $\epsilon = \rho = 1/2$, and indeed $4 = \width{\improvs} \ge 1 / \rho = 2$ and $1 = \width{\valids} \ge (1 - \epsilon) / \rho = 1$.
However, if we put $\rho = 1/3$ we violate the second inequality and the instance is not realizable: essentially, we need to distribute probability $1 - \epsilon = 1/2$ among plays in $\valids$ (to satisfy the soft constraint), but since $\width{\valids} = 1$, against some adversaries we can only generate one play in $\valids$ and would have to give it the whole $1/2$ (violating the randomness requirement).
\end{runningexample}

The difficult part of the Theorem is constructing an improviser when the inequalities (2) hold.
Despite the similarity in these conditions to the non-reactive case, the construction is much more involved.
We begin with a general overview.

\subsection{Improviser Construction: Discussion}

Our improviser can be viewed as an extension of the classical random-walk reduction of uniform sampling to counting \cite{wilf}.
In that algorithm (which was used in a similar way for DFA specifications in \cite{jacm,fsttcs}), a uniform distribution over paths in a DAG is obtained by moving to the next vertex with probability proportional to the number of paths originating at it.
In our case, which plays are possible depends on the adversary, but the width still tells us \emph{how many} plays are possible.
So we could try a random walk using widths as weights: e.g. on the first turn in Fig.~\ref{figure:running-game}, picking $+$, $-$, and $=$ with probabilities $1/4$, $2/4$, and $1/4$ respectively.
Against the adversary shown in Fig.~\ref{figure:running-game}, this would indeed yield a uniform distribution over the four possible plays in $\improvs$.

However, the soft constraint may require a non-uniform distribution.
In the running example with $\epsilon = \rho = 1/2$, we need to generate the single possible play in $\valids$ with probability $1/2$, not just the uniform probability $1/4$ .
This is easily fixed by doing the random walk with a \emph{weighted average} of the widths of $\improvs$ and $\valids$: specifically, move to position $h$ with probability proportional to $\alpha \widthrel{\valids}{h} + \beta (\widthrel{\improvs}{h} - \widthrel{\valids}{h})$.
In the example, this would result in plays in $\valids$ getting probability $\alpha$ and those in $\improvs \setminus \valids$ getting probability $\beta$.
Taking $\alpha$ sufficiently large, we can ensure the soft constraint is satisfied.

Unfortunately, this strategy can fail if the adversary makes \emph{more} plays available than the width guarantees.
Consider the game on the left of Fig.~\ref{figure:extra-paths-game}, where $\width{\improvs} = 3$ and $\width{\valids} = 2$.
This is realizable with $\epsilon = \rho = 1/3$, but no values of $\alpha$ and $\beta$ yield improvising strategies, essentially because an adversary moving from $X$ to $Z$ breaks the worst-case assumption that the adversary will minimize the number of possible plays by moving to $Y$.
In fact, this instance is realizable but not by any memoryless strategy.
To see this, note that all such strategies can be parametrized by the probabilities $p$ and $q$ in Fig.~\ref{figure:extra-paths-game}.
To satisfy the randomness constraint against the adversary that moves from $X$ to $Y$, both $p$ and $(1 - p) q$ must be at most $1/3$.
To satisfy the soft constraint against the adversary that moves from $X$ to $Z$ we must have $pq + (1 - p)q \ge 2/3$, so $q \ge 2/3$.
But then $(1 - p) q \ge (1 - 1/3) (2/3) = 4/9 > 1/3$, a contradiction.

\begin{figure}[tb]
\begin{minipage}{0.5\textwidth}
\centering
\begin{tikzpicture}[initial text=, transform shape, scale=0.7, node distance=0.5cm]

 \node[state, initial] (s0) {}; 
 \node[adversary, right= of s0, yshift=0.6cm] (s1) {\Large $X$}; 
 \node[adversary, right= of s0, yshift=-0.6cm] (s2) {};
 \node[accepting, soft, state, right= of s1, yshift=0.6cm] (s3) {\Large $Y$}; 
 \node[state, right= of s1, yshift=-0.6cm] (s6) {\Large $Z$};
 \node[accepting, soft, state, right= of s6, yshift=0.6cm] (s7) {\Large $W$};
 \node[accepting, state, right= of s6, yshift=-0.6cm] (s8) {};
 \node[right= of s2, yshift=-0.85cm] (spacer) {};

 % Direct 
 \path[->] 
 (s0) edge node[above, xshift=-0.1cm] {\Large $p$} (s1) 
 (s0) edge (s2)    
 (s1) edge (s3)
 (s1) edge (s6)
 (s2) edge (s6)
 (s6) edge node[above, xshift=-0.15cm] {\Large $q$} (s7)
 (s6) edge (s8);

\end{tikzpicture}
\end{minipage}
\begin{minipage}{0.5\textwidth}
\centering
\begin{tikzpicture}[initial text=, transform shape, scale=0.7, node distance=0.5cm]

 \node[state, initial] (s0) {\Large $P$}; 
 \node[adversary, right= of s0, yshift=0.6cm] (s1) {\Large $Q$}; 
 \node[adversary, right= of s0, yshift=-0.6cm] (s2) {\Large $R$};
 \node[accepting, soft, state, right= of s1, yshift=0.6cm] (s3) {}; 
 \node[accepting, soft, state, right= of s2, yshift=-0.6cm] (s4) {}; 
 \node[state, right= of s1, yshift=-0.6cm] (s6) {};
 \node[accepting, state, right= of s6, yshift=0.6cm] (s7) {};
 \node[accepting, state, right= of s6, yshift=-0.6cm] (s8) {};

 % Direct 
 \path[->] 
 (s0) edge (s1)
 (s0) edge (s2)    
 (s1) edge (s3)
 (s1) edge (s6)
 (s2) edge (s6)
 (s2) edge (s4)
 (s6) edge (s7)
 (s6) edge (s8);

\end{tikzpicture}
\end{minipage}
\caption{Reachability games where a na\"ive random walk, and all memoryless strategies, fail (left) and where no strategy can optimize either $\epsilon$ or $\rho$ against every adversary simultaneously (right).
}
\label{figure:extra-paths-game}
\label{figure:impossible-game}
\end{figure}
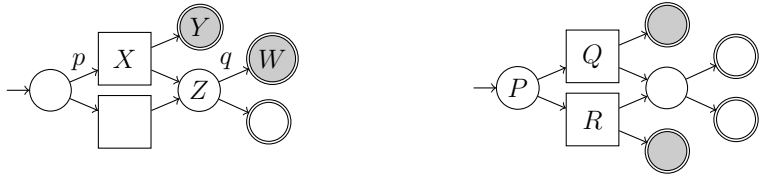

To fix this problem, our improvising strategy $\hat{\sigma}$ (which we will fully specify in Algorithm \ref{algorithm:impro} below) takes a simplistic approach: it tracks how many plays in $\valids$ and $\improvs$ are expected to be possible based on their widths, and if more are available it ignores them.
For example, entering state $Z$ from $X$ there are 2 ways to produce a play in $\improvs$, but since $\widthrel{\improvs}{X} = 1$ we ignore the play in $\improvs \setminus \valids$.
Extra plays in $\valids$ are similarly ignored by being treated as members of $\improvs \setminus \valids$.
Ignoring unneeded plays may seem wasteful, but the proof of Theorem \ref{theorem:feasibility} will show that $\hat{\sigma}$ nevertheless achieves the best possible $\epsilon$:
\begin{corollary} \label{corollary:optimal-epsilon}
$\mathcal{C}$ is realizable iff $\width{\improvs} \ge 1/\rho$ and $\epsilon \ge \eopt \equiv \max(1 - \rho \width{\valids}, 0)$.
Against any adversary, the error probability of Algorithm \ref{algorithm:impro} is at most $\eopt$.
\end{corollary}

Thus, if \emph{any} improviser can achieve an error probability $\epsilon$, ours does.
We could ask for a stronger property, namely that against each adversary the improviser achieves the smallest possible error probability \emph{for that adversary}.
Unfortunately, this is impossible in general.
Consider the game on the right in Fig.~\ref{figure:impossible-game}, with $\rho = 1$.
Against the adversary which always moves up, we can achieve $\epsilon = 0$ with the strategy that at $P$ moves to $Q$.
We can also achieve $\epsilon = 0$ against the adversary that always moves down, but only with a \emph{different} strategy, namely the one that at $P$ moves to $R$.
So there is no single strategy that achieves the optimal $\epsilon$ for every adversary.
A similar argument shows that there is also no strategy achieving the smallest possible $\rho$ for every adversary.
In essence, optimizing $\epsilon$ or $\rho$ in every case would require the strategy to depend on the adversary.

\subsection{Improviser Construction: Details}

Our improvising strategy, as outlined in the previous section, is shown in Algorithm~\ref{algorithm:impro}.
We first compute $\alpha$ and $\beta$, the (maximum) probabilities for generating elements of $\valids$ and $\improvs \setminus \valids$ respectively.
As in \cite{jacm}, we take $\alpha$ as large as possible given $\alpha \le \rho$, and determine $\beta$ from the probability left over (modulo a couple corner cases).

\begin{algorithm}[tb]
\caption{the strategy $\hat{\sigma}$}
\label{algorithm:impro}
\begin{algorithmic}[1]
\State $\alpha \gets \min(\rho, 1 / \width{\valids})$ \hspace{2.23cm} (or $0$ instead if $\width{\valids} = 0$)
\State $\beta \gets (1 - \alpha \width{\valids}) / (\width{\improvs} - \width{\valids})$ \quad (or $0$ instead if $\width{\improvs} - \width{\valids} = 0$)
\State $m^\valids \gets \width{\valids}$, $m^\improvs \gets \width{\improvs}$
\State $h \gets \lambda$
\While {the game is not over after $h$}
    \If {it is our turn after $h$}
    	\State $m^\valids_u, m^\improvs_u \gets \textsc{Partition}(m^\valids, m^\improvs, h)$ \Comment{returns values for each $u \in \alphabet$}
	\State for each $u \in \alphabet$, put $t_u \gets \alpha m^\valids_u + \beta (m^\improvs_u - m^\valids_u)$
	\State pick $u \in \alphabet$ with probability proportional to $t_u$ and append it to $h$
	\State $m^\valids \gets m^\valids_u$, $m^\improvs \gets m^\improvs_u$
    \Else
         \State the adversary picks $u \in \alphabet$ given the history $h$; append it to $h$
    \EndIf
\EndWhile
\Return $h$
\end{algorithmic}
\end{algorithm}

Next we initialize $m^\valids$ and $m^\improvs$, our expectations for how many plays in $\valids$ and $\improvs$ respectively are still possible to generate.
Initially these are given by $\width{\valids}$ and $\width{\improvs}$, but as we saw above it is possible for more plays to become available.
The function \textsc{Partition} handles this, deciding which $m^\valids$ (resp., $m^\improvs$) out of the available $\widthrel{\valids}{h}$ ($\widthrel{\improvs}{h}$) plays we will use.
The behavior of \textsc{Partition} is defined by the following lemma; its proof (in Appendix~\ref{section:detailed-proofs}) greedily takes the first $m^\valids$ possible plays in $\valids$ under some canonical order and the first $m^\improvs - m^\valids$ of the remaining plays in $\improvs$.

\begin{restatable}{lemma}{lemmaPartition} \label{lemma:partition}
If it is our turn after $h \in \histories$, and $m^\valids, m^\improvs \in \Z$ satisfy $0 \le m^\valids \le m^\improvs \le \widthrel{\improvs}{h}$ and $m^\valids \le \widthrel{\valids}{h}$, there are integer partitions $\sum_{u \in \alphabet} m^\valids_u$ and $\sum_{u \in \alphabet} m^\improvs_u$ of $m^\valids$ and $m^\improvs$ respectively such that $0 \le m^\valids_u \le m^\improvs_u \le \widthrel{\improvs}{hu}$ and $m^\valids_u \le \widthrel{\valids}{hu}$ for all $u \in \alphabet$.
These are computable in poly-time given oracles for $\widthrel{\improvs}{\cdot}$ and $\widthrel{\valids}{\cdot}$.
\end{restatable}

Finally, we perform the random walk, moving from position $h$ to $hu$ with (unnormalized) probability $t_u$, the weighted average described above.

\begin{runningexample}
With $\epsilon = \rho = 1/2$, as before $\width{\valids} = 1$ and $\width{\improvs} = 4$ so $\alpha = 1/2$ and $\beta = 1/6$.
On the first move, $m^\valids$ and $m^\improvs$ match $\widthrel{\valids}{h}$ and $\widthrel{\improvs}{h}$, so all plays are used and \textsc{Partition} returns $(\widthrel{\valids}{hu}, \widthrel{\improvs}{hu})$ for each $u \in \alphabet$.
Looking up these values in Fig.~\ref{figure:running-impro}, we see $(m^\valids_=, m^\improvs_=) = (0, 2)$ and so $t(=) = 2 \beta = 1/3$.
Similarly $t(+) = \alpha = 1/2$ and $t(-) = \beta = 1/6$.
We choose an action according to these weights; suppose $=$, so that we update $m^\valids \gets 0$ and $m^\improvs \gets 2$, and suppose the adversary responds with $=$.
From Fig.~\ref{figure:running-impro}, $\widthrel{\valids}{==} = 1$ and $\widthrel{\improvs}{==} = 3$, whereas $m^\valids = 0$ and $m^\improvs = 2$.
So \textsc{Partition} discards a play, say returning $(m^\valids_u, m^\improvs_u) = (0, 1)$ for $u \in \{ +, = \}$ and $(0, 0)$ for $u \in \{-\}$.
Then $t(+) = t(=) = \beta = 1/6$ and $t(-) = 0$.
So we pick $+$ or $=$ with equal probability, say $+$.
If the adversary responds with $+$, we get the play $=$$=$$+$$+$, shown in bold on Fig.~\ref{figure:running-impro}.
As desired, it satisfies the hard constraint.
\end{runningexample}

\begin{figure}[tb]
\begin{minipage}{0.5\textwidth}
\centering
\begin{tikzpicture}[initial text=, transform shape, scale=0.6, node distance=0.75cm]

 \node[state, initial] (s0) {\Large 4}; 
 \node[adversary, right= of s0, yshift=1.5cm] (s11) {\Large 1}; 
 \node[adversary, right= of s0] (s10) {\Large 2};
  \node[adversary, right= of s0, yshift=-1.5cm] (s1-1) {\Large 1};
 \node[state, right= of s11, yshift=1.5cm] (s22) {\Large 1};
 \node[state, right= of s11] (s21) {\Large 2};
 \node[state, right= of s10] (s20) {\Large 3};
 \node[state, right= of s1-1] (s2-1) {\Large 2};
 \node[state, right= of s1-1, yshift=-1.5cm] (s2-2) {\Large 1};
 \node[adversary, right= of s21, yshift=1.5cm] (s32) {\Large 0};
 \node[adversary, right= of s21] (s31) {\Large 1};
 \node[adversary, right= of s20] (s30) {\Large 1};
 \node[adversary, right= of s2-1] (s3-1) {\Large 1};
 \node[adversary, right= of s2-1, yshift=-1.5cm] (s3-2) {\Large 0};
 \node[state, accepting, soft, right= of s31, yshift=1.5cm] (s42) {\Large 1};
 \node[state, accepting, soft, right= of s31] (s41) {\Large 1};
 \node[state, accepting, soft, right= of s30] (s40) {\Large 1};
 \node[state, accepting, right= of s3-1] (s4-1) {\Large 1};
 \node[state, accepting, right= of s3-1, yshift=-1.5cm] (s4-2) {\Large 1};
 \node[right= of s22, yshift=1cm] (s33) {};
 \node[right= of s2-2, yshift=-1cm] (s3-3) {};
 \node[right= of s32, yshift=1cm] (s43) {};
 \node[right= of s3-2, yshift=-1cm] (s4-3) {};
 \node[right= of s42, xshift=-0.5cm] (l2) {\Large $+2$};
 \node[right= of s41, xshift=-0.5cm] (l1) {\Large $+1$};
 \node[right= of s40, xshift=-0.5cm] (l0) {\Large $+0$};
 \node[right= of s4-1, xshift=-0.5cm] (l-1) {\Large $-1$};
 \node[right= of s4-2, xshift=-0.5cm] (l-2) {\Large $-2$};
 \node[above= of l2, yshift=-0.5cm] (l3) {\Large $+3$};
 \node[below= of l-2, yshift=0.5cm] (l-3) {\Large $-3$};

 \path[->] 
 %(s0) edge (s11)
 %(s0) edge (s10)
 %(s0) edge (s1-1)
 (s11) edge (s22)
 (s11) edge (s21)
 (s11) edge (s20)
 (s10) edge (s21)
 %(s10) edge (s20)
 (s10) edge (s2-1)
 (s1-1) edge (s20)
 (s1-1) edge (s2-1)
 (s1-1) edge (s2-2)
 (s22) edge (s33)
 (s22) edge (s32)
 (s22) edge (s31)
 (s21) edge (s32)
 (s21) edge (s31)
 (s21) edge (s30)
 %(s20) edge (s31)
 (s20) edge (s30)
 (s20) edge (s3-1)
 (s2-1) edge (s30)
 (s2-1) edge (s3-1)
 (s2-1) edge (s3-2)
 (s2-2) edge (s3-1)
 (s2-2) edge (s3-2)
 (s2-2) edge (s3-3)
 (s32) edge (s43)
 (s32) edge (s42)
 (s32) edge (s41)
 %(s31) edge (s42)
 (s31) edge (s41)
 (s31) edge (s40)
 (s30) edge (s41)
 (s30) edge (s40)
 (s30) edge (s4-1)
 (s3-1) edge (s40)
 (s3-1) edge (s4-1)
 (s3-1) edge (s4-2)
 (s3-2) edge (s4-1)
 (s3-2) edge (s4-2)
 (s3-2) edge (s4-3);

 \draw[->]
 (s0) -- node[above, xshift=-0.2cm, yshift=0.1cm] {\Large $+$} (s11);
 \draw[->]
 (s0) -- node[below, xshift=-0.2cm, yshift=-0.1cm] {\Large $-$} (s1-1);
 \draw[->, line width=1.5pt]
 (s0) -- node[above, xshift=-0.1cm] {\Large $=$} (s10);
 \path[->, line width=1.5pt]
 (s10) edge (s20)
 (s20) edge (s31)
 (s31) edge (s42);
 
\end{tikzpicture}
\end{minipage}
\begin{minipage}{0.5\textwidth}
\centering
\begin{tikzpicture}[initial text=, transform shape, scale=0.6, node distance=0.75cm]

 \node[state, initial] (s0) {\Large 1}; 
 \node[adversary, right= of s0, yshift=1.5cm] (s11) {\Large 1}; 
 \node[adversary, right= of s0] (s10) {\Large 0};
  \node[adversary, right= of s0, yshift=-1.5cm] (s1-1) {\Large 0};
 \node[state, right= of s11, yshift=1.5cm] (s22) {\Large 1};
 \node[state, right= of s11] (s21) {\Large 1};
 \node[state, right= of s10] (s20) {\Large 1};
 \node[state, right= of s1-1] (s2-1) {\Large 0};
 \node[state, right= of s1-1, yshift=-1.5cm] (s2-2) {\Large 0};
 \node[adversary, right= of s21, yshift=1.5cm] (s32) {\Large 0};
 \node[adversary, right= of s21] (s31) {\Large 1};
 \node[adversary, right= of s20] (s30) {\Large 0};
 \node[adversary, right= of s2-1] (s3-1) {\Large 0};
 \node[adversary, right= of s2-1, yshift=-1.5cm] (s3-2) {\Large 0};
 \node[state, accepting, soft, right= of s31, yshift=1.5cm] (s42) {\Large 1};
 \node[state, accepting, soft, right= of s31] (s41) {\Large 1};
 \node[state, accepting, soft, right= of s30] (s40) {\Large 1};
 \node[state, accepting, right= of s3-1] (s4-1) {\Large 0};
 \node[state, accepting, right= of s3-1, yshift=-1.5cm] (s4-2) {\Large 0};
 \node[right= of s22, yshift=1cm] (s33) {};
 \node[right= of s2-2, yshift=-1cm] (s3-3) {};
 \node[right= of s32, yshift=1cm] (s43) {};
 \node[right= of s3-2, yshift=-1cm] (s4-3) {};

 \path[->] 
 (s0) edge (s11)
 %(s0) edge (s10)
 (s0) edge (s1-1)
 (s11) edge (s22)
 (s11) edge (s21)
 (s11) edge (s20)
 (s10) edge (s21)
 %(s10) edge (s20)
 (s10) edge (s2-1)
 (s1-1) edge (s20)
 (s1-1) edge (s2-1)
 (s1-1) edge (s2-2)
 (s22) edge (s33)
 (s22) edge (s32)
 (s22) edge (s31)
 (s21) edge (s32)
 (s21) edge (s31)
 (s21) edge (s30)
 %(s20) edge (s31)
 (s20) edge (s30)
 (s20) edge (s3-1)
 (s2-1) edge (s30)
 (s2-1) edge (s3-1)
 (s2-1) edge (s3-2)
 (s2-2) edge (s3-1)
 (s2-2) edge (s3-2)
 (s2-2) edge (s3-3)
 (s32) edge (s43)
 (s32) edge (s42)
 (s32) edge (s41)
 %(s31) edge (s42)
 (s31) edge (s41)
 (s31) edge (s40)
 (s30) edge (s41)
 (s30) edge (s40)
 (s30) edge (s4-1)
 (s3-1) edge (s40)
 (s3-1) edge (s4-1)
 (s3-1) edge (s4-2)
 (s3-2) edge (s4-1)
 (s3-2) edge (s4-2)
 (s3-2) edge (s4-3);

 \path[->, line width=1.5pt]
 (s0) edge (s10)
 (s10) edge (s20)
 (s20) edge (s31)
 (s31) edge (s42);

\end{tikzpicture}
\end{minipage}
\caption{A run of Algorithm~\ref{algorithm:impro}, labeling states with corresponding widths of $\improvs$ (left) and $\valids$ (right).}
\label{figure:running-impro}
\end{figure}
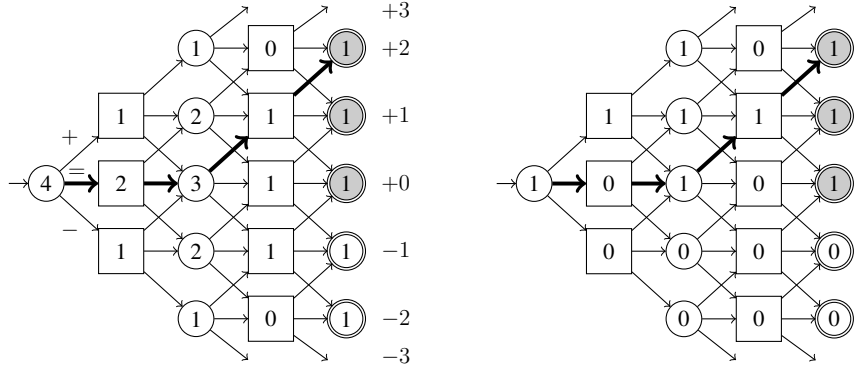

The next few lemmas establish that $\hat{\sigma}$ is well-defined and in fact an improvising strategy, allowing us to prove Theorem~\ref{theorem:feasibility}.
Throughout, we write $m^\valids(h)$ (resp., $m^\improvs(h)$) for the value of $m^\valids$ ($m^\improvs$) at the start of the iteration for history $h$.
We also write $t(h) = \alpha m^\valids(h) + \beta (m^\improvs(h) - m^\valids(h))$ (so $t(hu) = t_u$ when we pick $u$).

\begin{restatable}{lemma}{lemmaSigmaHard} \label{lemma:sigma-hard}
If $\width{\improvs} \ge 1 / \rho$, then $\hat{\sigma}$ is a well-defined strategy and $P_{\hat{\sigma},\tau}(\improvs) = 1$ for every adversary $\tau$.
\end{restatable}
\begin{proof}[sketch]
An easy induction on $h$ shows the conditions of Lemma \ref{lemma:partition} are always satisfied, and that $t(h)$ is always positive since we never pick a $u$ with $t_u = 0$.
So $\sum_u t_u = t(h) > 0$ and $\hat{\sigma}$ is well-defined.
Furthermore, $t(h) > 0$ implies $m^\improvs(h) > 0$, so for any $h \in \possible_{\hat{\sigma},\tau}$ we have $\ind{h \in \improvs} = \widthrel{\improvs}{h} \ge m^\improvs(h) > 0$ and thus $h \in \improvs$. \qed
\end{proof}

\begin{restatable}{lemma}{lemmaSigmaSoft} \label{lemma:sigma-soft}
If $\width{\improvs} \ge 1 / \rho$, then $P_{\hat{\sigma},\tau}(\valids) \ge \min(\rho \width{\valids}, 1)$ for every $\tau$.
\end{restatable}
\begin{proof}[sketch]
Because of the $\alpha m^\valids(h)$ term in the weights $t(h)$, the probability of obtaining a play in $\valids$ starting from $h$ is at least $\alpha m^\valids(h) / t(h)$ (as can be seen by induction on $h$ in order of decreasing length).
Then since $m^\valids(\lambda) = \width{\valids}$ and $t(\lambda) = 1$ we have $P_{\hat{\sigma},\tau}(\valids) \ge \alpha \width{\valids} = \min(\rho \width{\valids}, 1)$. \qed
\end{proof}

\begin{restatable}{lemma}{lemmaSigmaRandom} \label{lemma:sigma-random}
If $\width{\improvs} \ge 1 / \rho$, then $P_{\hat{\sigma},\tau}(\pi) \le \rho$ for every $\pi \in \alphabet^n$ and $\tau$.
\end{restatable}
\begin{proof}[sketch]
If the adversary is deterministic, the weights we use for our random walk yield a distribution where each play $\pi$ has probability either $\alpha$ or $\beta$ (depending on whether $m^\valids(\pi) = 1$ or $0$).
If the adversary assigns nonzero probability to multiple choices this only decreases the probability of individual plays.
Finally, since $\width{\improvs} \ge 1 / \rho$ we have $\alpha, \beta \le \rho$. \qed
\end{proof}

\begin{proof}[of Theorem \ref{theorem:feasibility}]
We use a similar argument to that of \cite{jacm}.
\begin{description}
\item[(1)$\Rightarrow$(2)] Suppose $\sigma$ is an improvising strategy, and fix any adversary $\tau$.
Then $\rho |\possible_{\sigma,\tau} \cap \improvs| = \sum_{\pi \in \possible_{\sigma,\tau} \cap \improvs} \rho \ge \sum_{\pi \in \improvs} P_{\sigma,\tau}(\pi) = P_{\sigma,\tau}(\improvs) = 1$, so $|\possible_{\sigma,\tau} \cap \improvs| \ge 1/\rho$.
Since $\tau$ is arbitrary, this implies $\width{\improvs} \ge 1/\rho$.
Since $\valids \subseteq \improvs$, we also have $\rho |\possible_{\sigma,\tau} \cap \valids| = \sum_{\pi \in \possible_{\sigma,\tau} \cap \valids} \rho \ge \sum_{\pi \in \valids} P_{\sigma,\tau}(\pi) = P_{\sigma,\tau}(\valids) \ge 1 - \epsilon$, so $|\possible_{\sigma,\tau} \cap \valids| \ge (1-\epsilon)/\rho$ and thus $\width{\valids} \ge (1-\epsilon)/\rho$.

\item[(2)$\Rightarrow$(3)] By Lemmas \ref{lemma:sigma-hard} and \ref{lemma:sigma-random}, $\hat{\sigma}$ is well-defined and satisfies the hard and randomness constraints.
By Lemma \ref{lemma:sigma-soft}, $P_{\hat{\sigma},\tau}(\valids) \ge \min(\rho \width{\valids}, 1) \ge 1 - \epsilon$, so $\hat{\sigma}$ also satisfies the soft constraint and thus is an improvising strategy.
Its transition probabilities are rational, so it can be implemented by an expected finite-time probabilistic algorithm, which is then an improviser for $\mathcal{C}$.

\item[(3)$\Rightarrow$(1)] Immediate. \qed
\end{description}
\end{proof}

\begin{proof}[of Corollary \ref{corollary:optimal-epsilon}]
The inequalities in the statement are equivalent to those of Theorem \ref{theorem:feasibility}(2).
By Lemma \ref{lemma:sigma-soft}, we have $P_{\hat{\sigma},\tau}(\valids) \ge \min(\rho \width{\valids}, 1)$.
So the error probability is at most $1 - \min(\rho \width{\valids}, 1) = \eopt$. \qed
\end{proof}

\section{A Generic Improviser} \label{section:generic}

We now use the construction of Sec.~\ref{section:existence} to develop a generic improvisation scheme usable with any class of specifications \textsc{Spec} supporting the following operations:
\begin{description}
\item[Intersection:] Given specs $\mathcal{X}$ and $\mathcal{Y}$, find $\mathcal{Z}$ such that $L(\mathcal{Z}) = L(\mathcal{X}) \cap L(\mathcal{Y})$. \smallskip
\item[Width Measurement:] Given a specification $\mathcal{X}$, a length $n \in \N$ in unary,  and a history $h \in \histories$, compute $\widthrel{X}{h}$ where $X = L(\mathcal{X}) \cap \alphabet^n$.
\end{description}

Efficient algorithms for these operations lead to efficient improvisation schemes:
\begin{theorem} \label{theorem:generic-scheme}
If the operations on \textsc{Spec} above take polynomial time (resp. space), then $\rcic{\textsc{Spec}}{\textsc{Spec}}$ has a polynomial-time (space) improvisation scheme.
\end{theorem}
\begin{proof}
Given an instance $\mathcal{C} = (\hard, \soft, n, \epsilon, \rho)$ in $\rcic{\textsc{Spec}}{\textsc{Spec}}$, we first apply intersection to $\hard$ and $\soft$ to obtain $\mathcal{\valids} \in \textsc{Spec}$ such that $L(\mathcal{\valids}) \cap \alphabet^n = \valids$.
Since intersection takes polynomial time (space), $\mathcal{\valids}$ has size polynomial in $|\mathcal{C}|$.
Next we use width measurement to compute $\width{\improvs} = \widthrel{L(\hard) \cap \alphabet^n}{\emptyword}$ and $\width{\valids} = \widthrel{L(\mathcal{\valids}) \cap \alphabet^n}{\emptyword}$.
If these violate the inequalities in Theorem \ref{theorem:feasibility}, then $\mathcal{C}$ is not realizable and we return $\bot$.
Otherwise $\mathcal{C}$ is realizable, and $\hat{\sigma}$ above is an improvising strategy.
Furthermore, we can construct an expected finite-time probabilistic algorithm implementing $\hat{\sigma}$, using width measurement to instantiate the oracles needed by Lemma \ref{lemma:partition}.
Determining $m^\valids(h)$ and $m^\improvs(h)$ takes $O(n)$ invocations of \textsc{Partition}, each of which is poly-time relative to the width measurements.
These take time (space) polynomial in $|\mathcal{C}|$, since $\hard$ and $\mathcal{\valids}$ have size polynomial in $|\mathcal{C}|$.
As $m^\valids, m^\improvs \le |\alphabet|^n$, they have polynomial bitwidth and so the arithmetic required to compute $t_u$ for each $u \in \alphabet$ takes polynomial time.
Therefore the total expected runtime (space) of the improviser is polynomial. \qed
\end{proof}

Note that as a byproduct of testing the inequalities in Theorem \ref{theorem:feasibility}, our algorithm can compute the best possible error probability $\eopt$ given $\hard$, $\soft$, and $\rho$ (see Corollary \ref{corollary:optimal-epsilon}).
Alternatively, given $\epsilon$, we can compute the best possible $\rho$.

We will see below how to efficiently compute widths for DFAs, so Theorem~\ref{theorem:generic-scheme} yields a polynomial-time improvisation scheme.
If we allow polynomial-\emph{space} schemes, we can use a general technique for width measurement that only requires a very weak assumption on the specifications, namely testability in polynomial space:
\begin{restatable}{theorem}{theoremPspaceScheme} \label{theorem:pspace-scheme}
$\rcic{\textsc{PSA}}{\textsc{PSA}}$ has a polynomial-space improvisation scheme, where \textsc{PSA} is the class of polynomial-space decision algorithms.
\end{restatable}
\begin{proof}[sketch]
We apply Theorem~\ref{theorem:generic-scheme}, computing widths recursively using Lemma~\ref{lemma:widthrel}, (3)--(5).
As in the $\PSPACE$ $\QBF$ algorithm, the current path in the recursive tree and required auxiliary storage need only polynomial space. \qed
\end{proof}

\section{Reachability Games and DFAs} \label{section:reachability}

Now we develop a polynomial-time improvisation scheme for RCI instances with $\DFA$ specifications.
This also provides a scheme for reachability/safety games, whose winning conditions can be straightforwardly encoded as $\DFA$s.

Suppose $D$ is a $\DFA$ with states $V$, accepting states $T$, and transition function $\delta : V \times \alphabet \rightarrow V$.
Our scheme is based on the fact that $\widthrel{L(D)}{h}$ depends only on the state of $D$ reached on input $h$, allowing these widths to be computed by dynamic programming.
Specifically, for all $v \in V$ and $i \in \{0, \dots, n\}$ we define:
\[
\pweight{v}{i} =
\begin{cases}
\ind{v \in T} & i = n \\
\min_{u \in \alphabet} \; \pweight{\delta(v,u)}{i+1} & i < n \land i \text{ odd} \\
\sum_{u \in \alphabet} \; \pweight{\delta(v,u)}{i+1}  & \text{otherwise}.
\end{cases}
\]

\begin{runningexample}
Figure~\ref{figure:running-counts} shows the values $C(v,i)$ in rows from $i=n$ downward.
For example, $i=2$ is our turn, so $C(1,2) = C(0,3) + C(1, 3) + C(2, 3) = 1 + 1 + 0 = 2$, while $i = 3$ is the adversary's turn, so $C(-3, 3) = \min \{ C(-3, 4) \} = \min \{ 0 \} = 0$.
Note that the values in Fig.~\ref{figure:running-counts} agree with the widths $\widthrel{\improvs}{h}$ shown in Fig.~\ref{figure:running-impro}.
\end{runningexample}

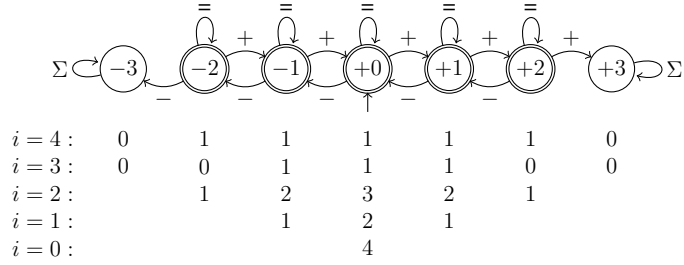
\begin{figure}[tb]
\centering
\begin{tikzpicture}[initial text=, transform shape, scale=0.6, node distance=0.75cm]

 \node[state, accepting, initial below] (s0) {\Large $+0$}; 
 \node[state, accepting, right= of s0] (s1) {\Large $+1$}; 
 \node[state, accepting, right= of s1] (s2) {\Large $+2$}; 
 \node[state, right= of s2] (s3) {\Large $+3$}; 
 \node[state, accepting, left= of s0] (s-1) {\Large $-1$}; 
 \node[state, accepting, left= of s-1] (s-2) {\Large $-2$}; 
 \node[state, left= of s-2] (s-3) {\Large $-3$}; 
 
 \node[below= of s0] (w40) {\Large $1$};
 \node[below= of w40, yshift=0.7cm] (w30) {\Large $1$};
 \node[below= of w30, yshift=0.7cm] (w20) {\Large $3$};
 \node[below= of w20, yshift=0.7cm] (w10) {\Large $2$};
 \node[below= of w10, yshift=0.7cm] (w00) {\Large $4$};
 \node[below= of s-1] (w4-1) {\Large $1$};
 \node[below= of w4-1, yshift=0.7cm] (w3-1) {\Large $1$};
 \node[below= of w3-1, yshift=0.7cm] (w2-1) {\Large $2$};
 \node[below= of w2-1, yshift=0.7cm] (w1-1) {\Large $1$};
 \node[below= of s1] (w41) {\Large $1$};
 \node[below= of w41, yshift=0.7cm] (w31) {\Large $1$};
 \node[below= of w31, yshift=0.7cm] (w21) {\Large $2$};
 \node[below= of w21, yshift=0.7cm] (w11) {\Large $1$};
 \node[below= of s-2] (w4-2) {\Large $1$};
 \node[below= of w4-2, yshift=0.7cm] (w3-2) {\Large $0$};
 \node[below= of w3-2, yshift=0.7cm] (w2-2) {\Large $1$};
 \node[below= of s2] (w42) {\Large $1$};
 \node[below= of w42, yshift=0.7cm] (w32) {\Large $0$};
 \node[below= of w32, yshift=0.7cm] (w22) {\Large $1$};
 \node[below= of s-3] (w4-3) {\Large $0$};
 \node[below= of w4-3, yshift=0.7cm] (w3-2) {\Large $0$};
 \node[below= of s3] (w43) {\Large $0$};
 \node[below= of w43, yshift=0.7cm] (w32) {\Large $0$};
 
 \node[left= of w4-3] (i4) {\Large $i=4:$};
 \node[below= of i4, yshift=0.7cm] (i3) {\Large $i=3:$};
 \node[below= of i3, yshift=0.7cm] (i2) {\Large $i=2:$};
 \node[below= of i2, yshift=0.7cm] (i1) {\Large $i=1:$};
 \node[below= of i1, yshift=0.7cm] (i0) {\Large $i=0:$};
 
 \path[->]
 (s-2) edge[bend left] node[above] {\Large $+$} (s-1)
 (s-1) edge[bend left] node[above] {\Large $+$} (s0)
 (s0) edge[bend left] node[above] {\Large $+$} (s1)
 (s1) edge[bend left] node[above] {\Large $+$}(s2)
 (s2) edge[bend left] node[above] {\Large $+$} (s3)
 (s2) edge[bend left] node[below] {\Large $-$} (s1)
 (s1) edge[bend left] node[below] {\Large $-$} (s0)
 (s0) edge[bend left] node[below] {\Large $-$} (s-1)
 (s-1) edge[bend left] node[below] {\Large $-$}(s-2)
 (s-2) edge[bend left] node[below] {\Large $-$} (s-3)
 (s-3) edge[loop left] node[left] {\Large $\alphabet$} (s-3)
 (s-2) edge[loop above] node[above] {\Large =} (s-2)
 (s-1) edge[loop above] node[above] {\Large =} (s-1)
 (s0) edge[loop above] node[above] {\Large =} (s0)
 (s1) edge[loop above] node[above] {\Large =} (s1)
 (s2) edge[loop above] node[above] {\Large =} (s2)
 (s3) edge[loop right] node[right] {\Large $\alphabet$} (s3);
 
 \end{tikzpicture}
\caption{The hard specification DFA $\hard$ in our running example, showing how $\widthrel{\improvs}{h}$ is computed.}
\label{figure:running-counts}
\end{figure}

\begin{lemma} \label{lemma:width-p}
For any history $h \in \histories$, writing $X = L(D) \cap \alphabet^n$ we have $\widthrel{X}{h} = \pweight{D(h)}{|h|}$, where $D(h)$ is the state reached by running $D$ on $h$.
\end{lemma}
\begin{proof}
We prove this by induction on $i = |h|$ in decreasing order.
In the base case $i = n$, we have $\widthrel{X}{h} = \ind{h \in X} = \ind{D(h) \in T} = \pweight{D(h)}{n}$.
Now take any history $h \in \histories$ with $|h| = i < n$.
By hypothesis, for any $u \in \alphabet$ we have $\widthrel{X}{hu} = \pweight{D(hu)}{i+1}$.
If it is our turn after $h$, then $\widthrel{X}{h} = \sum_{u \in \alphabet} \widthrel{X}{hu} = \sum_{u \in \alphabet} \pweight{D(hu)}{i+1} = \pweight{D(h)}{i}$ as desired.
If instead it is the adversary's turn after $h$, then $\widthrel{X}{h} = \min_{u \in \alphabet} \widthrel{X}{hu} = \min_{u \in \alphabet} \pweight{D(hu)}{i+1} = \pweight{D(h)}{i}$ again as desired.
So by induction the hypothesis holds for any $i$. \qed
\end{proof}

\begin{theorem}
$\rcic{\DFA}{\DFA}$ has a polynomial-time improvisation scheme.
\end{theorem}
\begin{proof}
We implement Theorem \ref{theorem:generic-scheme}.
Intersection can be done with the standard product construction.
For width measurement we compute the quantities $\pweight{v}{i}$ by dynamic programming (from $i = n$ down to $i = 0$) and apply Lemma \ref{lemma:width-p}. \qed
\end{proof}

\section{Temporal Logics and Other Specifications} \label{section:temporal}

In this section we analyze the complexity of reactive control improvisation for specifications in the popular temporal logics $\LTL$ and $\LDL$.
We also look at $\NFA$ and $\CFG$ specifications, previously studied for non-reactive CI \cite{jacm}, to see how their complexities change in the reactive case.

For $\LTL$ specifications, reactive control improvisation is $\PSPACE$-hard because this is already true of ordinary reactive synthesis in a finite window (we suspect this has been observed but could not find a proof in the literature).
\begin{restatable}{theorem}{theoremLTLHardness} \label{theorem:ltl-hardness}
Finite-window reactive synthesis for $\LTL$ is $\PSPACE$-hard.
\end{restatable}
\begin{proof}[sketch]
Given a $\QBF$ $\phi = \exists x \forall y \dots \chi$, we can view assignments to its variables as traces over a single proposition.
In polynomial time we can construct an LTL formula $\psi$ whose models are the satisfying assignments of $\chi$.
Then there is a winning strategy to generate a play satisfying $\psi$ iff $\phi$ is true. \qed
\end{proof}

\begin{corollary}
$\rcic{\LTL}{\alphabet^*}$ and $\rcic{\alphabet^*}{\LTL}$ are $\PSPACE$-hard.
\end{corollary}
This is perhaps disappointing, but is an inevitable consequence of $\LTL$ subsuming Boolean formulas.
On the other hand, our general polynomial-space scheme applies to $\LTL$ and its much more expressive generalization $\LDL$:

\begin{theorem}
$\rcic{\LDL}{\LDL}$ has a polynomial-space improvisation scheme.
\end{theorem}
\begin{proof}
This follows from Theorem \ref{theorem:pspace-scheme}, since satisfaction of an $\LDL$ formula by a finite word can be checked in polynomial time (e.g. by combining dynamic programming on subformulas with a regular expression parser). \qed
\end{proof}
Thus for temporal logics polynomial-time algorithms are unlikely, but adding randomization to reactive synthesis does not increase its complexity.

The same is true for $\NFA$ and $\CFG$ specifications, where it is again $\PSPACE$-hard to find even a single winning strategy:
\begin{restatable}{theorem}{theoremNFA}
Finite-window reactive synthesis for $\NFA$s is $\PSPACE$-hard.
\end{restatable}
\begin{proof}[sketch]
Reduce from $\QBF$ as in Theorem~\ref{theorem:ltl-hardness}, constructing an NFA accepting the satisfying assignments of $\chi$ (as done in \cite{sharpNFA}). \qed
\end{proof}

\begin{corollary}
$\rcic{\NFA}{\alphabet^*}$ and $\rcic{\alphabet^*}{\NFA}$ are $\PSPACE$-hard.
\end{corollary}

\begin{theorem}
$\rcic{\CFG}{\CFG}$ has a polynomial-space improvisation scheme.
\end{theorem}
\begin{proof}
By Theorem \ref{theorem:pspace-scheme}, since $\CFG$ parsing can be done in polynomial time. \qed
\end{proof}

Since $\NFA$s can be converted to $\CFG$s in polynomial time, this completes the picture for the kinds of CI specifications previously studied.
In non-reactive CI, $\DFA$ specifications admit a polynomial-time improvisation scheme while for $\NFA$s/$\CFG$s the CI problem is $\sharpP$-equivalent \cite{jacm}.
Adding reactivity, $\DFA$ specifications remain polynomial-time while $\NFA$s and $\CFG$s move up to $\PSPACE$.

\section{Conclusion} \label{section:conclusion}

In this paper we introduced \emph{reactive control improvisation} as a framework for modeling reactive synthesis problems where random but controlled behavior is desired.
RCI provides a natural way to tune the amount of randomness while ensuring that safety or other constraints remain satisfied.
We showed that RCI problems can be efficiently solved in many cases occurring in practice, giving a polynomial-time improvisation scheme for reachability/safety or DFA specifications.
We also showed that RCI problems with specifications in $\LTL$ or $\LDL$, popularly used in planning, have the $\PSPACE$-hardness typical of bounded games, and gave a matching polynomial-space improvisation scheme.
This scheme generalizes to any specification checkable in polynomial space, including NFAs, CFGs, and many more expressive formalisms.
Table \ref{table:complexities} summarizes these results.

\begin{table}[tb]
\caption{Complexity of the reactive control improvisation problem for various types of hard and soft specifications $\hard$, $\soft$.
Here $\PSPACE$ indicates that checking realizability is $\PSPACE$-hard, and that there is a polynomial-space improvisation scheme.\label{table:complexities}}
\begin{center}
{
\setlength{\tabcolsep}{5pt}
\renewcommand{\arraystretch}{1.2}
\begin{tabular}{c||c|c|c|c|c|c|}
$\hard \backslash \soft$ & \textbf{RSG} & \textbf{DFA} & \textbf{NFA} & \textbf{CFG} & \textbf{LTL} & \textbf{LDL} \\
\hline
\hline
\textbf{RSG} & \multicolumn{2}{c|}{\multirow{2}{*}{poly-time}} & \multicolumn{4}{c|}{} \\
\cline{1-1}
\textbf{DFA} & \multicolumn{2}{c|}{} & \multicolumn{4}{c|}{} \\
\cline{1-3}
\textbf{NFA} & \multicolumn{6}{c|}{} \\
\cline{1-1}
\textbf{CFG} & \multicolumn{2}{c}{} & \multicolumn{4}{c|}{\multirow{2}{*}{\PSPACE}} \\
\cline{1-1}
\textbf{LTL} & \multicolumn{6}{r|}{}  \\
\cline{1-1}
\textbf{LDL} & \multicolumn{6}{r|}{} \\
\hline
\end{tabular}
}
\end{center}
\vspace{-0.4cm}
\end{table}

These results show that, at a high level, finding a maximally-randomized strategy using RCI is no harder than finding any winning strategy at all: for specifications yielding games solvable in polynomial time (respectively, space), we gave polynomial-time (space) improvisation schemes.
We therefore hope that in applications where ordinary reactive synthesis has proved tractable, our notion of randomized reactive synthesis will also.
In particular, we expect our DFA scheme to be quite practical, and are experimenting with applications in robotic planning.
On the other hand, our scheme for temporal logic specifications seems unlikely to be useful in practice without further refinement.
An interesting direction for future work would be to see if modern solvers for quantified Boolean formulas (QBF) could be leveraged or extended to solve these RCI problems.
This could be useful even for DFA specifications, as conjoining many simple properties can lead to exponentially-large automata.
Symbolic methods based on constraint solvers would avoid such blow-up.

We are also interested in extending the RCI problem definition to unbounded or infinite words, as typically used in reactive synthesis.
These extensions, as well as that to continuous signals, would be useful in robotic planning, cyber-physical system testing, and other applications.
However, it is unclear how best to adapt our randomness constraint to settings where the improviser can generate infinitely many words.
In such settings the improviser could assign arbitrarily small or even zero probability to every word, rendering the randomness constraint trivial.
Even in the bounded case, RCI extensions with more complex randomness constraints than a simple upper bound on individual word probabilities would be worthy of study.
One possibility would be to more directly control diversity and/or unpredictability by requiring the distribution of the improviser's output to be close to uniform after transformation by a given function.\\

\noindent \textbf{Acknowledgements.}
The authors would like to thank Markus Rabe, Moshe Vardi, and several anonymous reviewers for helpful discussions and comments, and Ankush Desai and Tommaso Dreossi for assistance with the drone simulations.
This work is supported in part by the National Science Foundation Graduate Research Fellowship Program under Grant No.~DGE-1106400, by NSF grants CCF-1139138 and CNS-1646208, by DARPA under agreement number FA8750-16-C0043, and by TerraSwarm, one of six centers of STARnet, a Semiconductor Research Corporation program sponsored by MARCO and DARPA.

\bibliographystyle{splncs03}
\bibliography{main}

\newpage
\appendix

\section{Patrolling Drone Experiments} \label{section:drone-experiments}

As described above, we ran experiments with two adversary strategies: one that moves towards the patrolling drone whenever possible, and one that moves in a fixed loop.
We ran the improviser four times against each adversary, obtaining the trajectories in Figures~\ref{figure:drones-chase} and \ref{figure:drones-loop}.
Animations showing the trajectories over time (and so illustrating that collisions do not in fact occur) are available online~\cite{videos}.
This site also provides our implementation of the DFA improvisation scheme, and implementations of the specifications and adversaries used in our drone experiments (as well as an adversary controlled by the user, so that one can type in actions and see how the improviser responds).

\begin{figure}[htb]
\centering
\includegraphics[width=\textwidth]{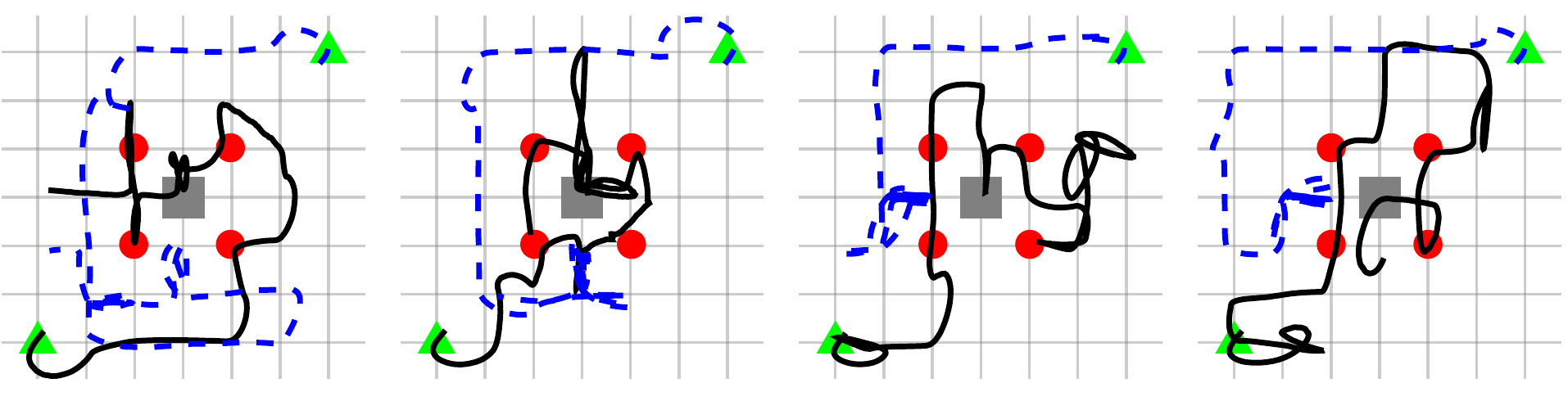}
\caption{Improvised trajectories against an adversary which moves toward the patroller when possible.}
\label{figure:drones-chase}
\end{figure}

\begin{figure}[htb]
\centering
\includegraphics[width=\textwidth]{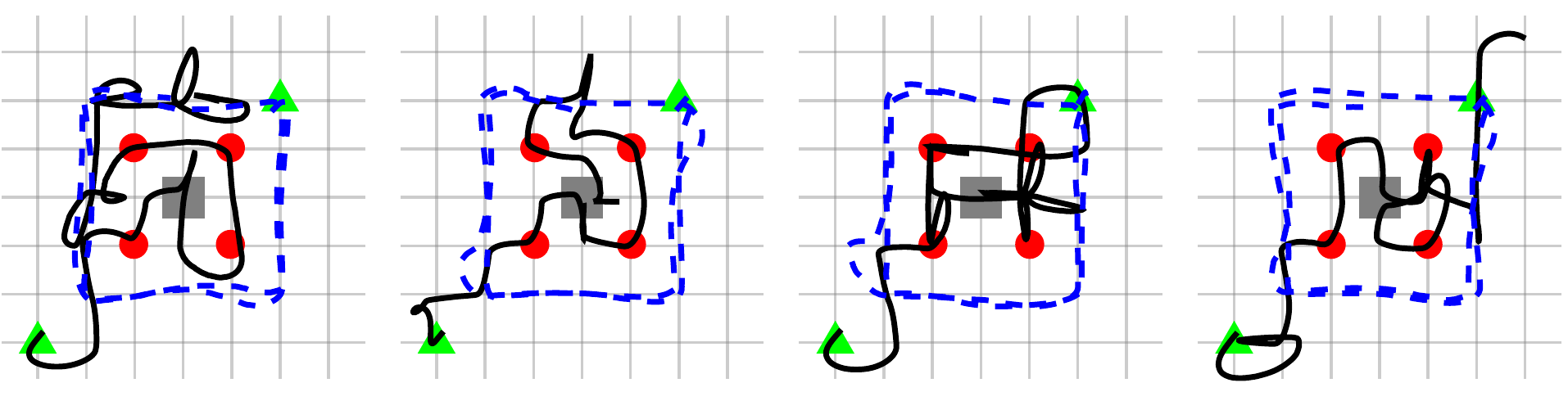}
\caption{Improvised trajectories against an adversary which moves in a loop.}
\label{figure:drones-loop}
\end{figure}

\section{Detailed Proofs} \label{section:detailed-proofs}

We use without comment several basic facts about $P_{\sigma,\tau}(\rho | h)$, all immediate from its definition:
\begin{lemma}
For any history $h \in \histories$, word $\rho \in \alphabet^{n - |h|}$, and strategies $\sigma$, $\tau$:
\begin{enumerate}[(1)]
 \item if $|h| = 0$, then $P_{\sigma,\tau}(\rho | h) = P_{\sigma,\tau}(\rho)$;
 \item if $|h| = n$, then $P_{\sigma,\tau}(\rho | h) = 1$;
 \item if $|h| < n$, then $\rho = u \rho'$ for some $u \in \alphabet$, and:
 \begin{enumerate}
 \item if it is our turn after $h$, then $P_{\sigma,\tau}(\rho | h) = \sigma(h, u) \cdot P_{\sigma,\tau}(\rho' | h u)$;
 \item if it is the adversary's turn after $h$, then $P_{\sigma,\tau}(\rho | h) = \tau(h, u) \cdot P_{\sigma,\tau}(\rho' | h u)$.
 \end{enumerate}
\end{enumerate}
\end{lemma}

\lemmaWidthrel*
\begin{proof}
\mbox{}
\begin{enumerate}
\item By definition, $\widthrel{X}{h} = \max_\sigma \min_\tau | \{ \pi \st h \pi \in X \land P_{\sigma,\tau}(\pi | h) > 0 \} |$, so $\widthrel{X}{h} \ge 0$ trivially.
Since $X \subseteq \alphabet^n$, if $h \pi \in X$ then $\pi \in \alphabet^{n - |h|}$.
So $\widthrel{X}{h} \le |\alphabet|^{n - |h|}$.
\item $\widthrel{X}{\emptyword} = \max_\sigma \min_\tau | \{ \pi \st \pi \in X \land P_{\sigma,\tau}(\pi) > 0 \} | = \max_\sigma \min_\tau | X \cap \possible_{\sigma,\tau} | = \width{X}$.
\item If $h \in X$, then the only word of the form $h \pi$ in $X$ is $h$, with $\pi = \emptyword$ (and $P_{\sigma,\tau}(\lambda | h) = 1 > 0$). Otherwise there is no word of the form $h \pi$ in $X$.
\item Since it is our turn, and in particular the game has not ended, every play $h \pi$ that can be generated given history $h$ has the form $h u \pi'$ for some $u \in \alphabet$.
So for any strategies $\sigma$ and $\tau$ we have $| \{ \pi \st h \pi \in X \land P_{\sigma,\tau}(\pi | h) > 0 \} | = \sum_{u \in \alphabet} | \{ \pi' \st h u \pi' \in X \land P_{\sigma,\tau}(u \pi' | h) > 0 \} | = \sum_{u \in \alphabet} | \{ \pi' \st h u \pi' \in X \land \sigma(h, u) \cdot P_{\sigma,\tau}(\pi' | h u) > 0 \} |$.

\quad For each $u \in \alphabet$, let $\sigma_u$ be a strategy witnessing $\widthrel{X}{hu}$.
Let $\tilde{\sigma}$ be a strategy which on history $h$ picks $u \in \alphabet$ uniformly at random, and on histories prefixed by $hu$ follows $\sigma_u$ (otherwise picking arbitrarily).
Then
\begin{align*}
\widthrel{X}{h} &\ge \min_\tau | \{ \pi \st h \pi \in X \land P_{\tilde{\sigma},\tau}(\pi | h) > 0 \} | \\
&= \min_\tau \sum_{u \in \alphabet} | \{ \pi' \st h u \pi' \in X \land \tilde{\sigma}(h,u) \cdot P_{\tilde{\sigma},\tau}(\pi' | h u) > 0 \} | \\
&= \min_\tau \sum_{u \in \alphabet} | \{ \pi' \st h u \pi' \in X \land P_{\tilde{\sigma},\tau}(\pi' | h u) > 0 \} | \\
&\ge \sum_{u \in \alphabet} \min_\tau | \{ \pi' \st h u \pi' \in X \land P_{\tilde{\sigma},\tau}(\pi' | h u) > 0 \} | \\
&= \sum_{u \in \alphabet} \min_\tau | \{ \pi' \st h u \pi' \in X \land P_{\sigma_u,\tau}(\pi' | h u) > 0 \} | \\
&= \sum_{u \in \alphabet} \widthrel{X}{hu} .
\end{align*}

For the other direction, let $\tilde{\sigma}$ be a strategy witnessing $\widthrel{X}{h}$, and for each $u \in \alphabet$ let $\tau_u = \arg\min_\tau | \{ \pi' \st h u \pi' \in X \land P_{\tilde{\sigma},\tau}(\pi' | h u) > 0 \} |$.
Let $\tilde{\tau}$ be a strategy which on histories prefixed by $h u$ follows $\tau_u$ (otherwise picking arbitrarily).
Then
\begin{align*}
\widthrel{X}{h} &\le | \{ \pi \st h \pi \in X \land P_{\tilde{\sigma},\tilde{\tau}}(\pi | h) > 0 \} | \\
&= \sum_{u \in \alphabet} | \{ \pi' \st h u \pi' \in X \land \tilde{\sigma}(h, u) \cdot P_{\tilde{\sigma},\tilde{\tau}}(\pi' | h u) > 0 \} | \\
&\le \sum_{u \in \alphabet} | \{ \pi' \st h u \pi' \in X \land P_{\tilde{\sigma},\tilde{\tau}}(\pi' | h u) > 0 \} | \\
&= \sum_{u \in \alphabet} | \{ \pi' \st h u \pi' \in X \land P_{\tilde{\sigma},\tau_u}(\pi' | h u) > 0 \} | \\
&\le \sum_{u \in \alphabet} \widthrel{X}{hu}.
\end{align*}
\item Since it is the adversary's turn (and in particular the game has not ended), for any strategies $\sigma$ and $\tau$ we have $| \{ \pi \st h \pi \in X \land P_{\sigma,\tau}(\pi | h) > 0 \} | = \sum_{u \in \alphabet} | \{ \pi' \st h u \pi' \in X \land P_{\sigma,\tau}(u \pi' | h) > 0 \} | = \sum_{u \in \alphabet} | \{ \pi' \st h u \pi' \in X \land \tau(h, u) \cdot P_{\sigma,\tau}(\pi' | h u) > 0 \} |$.

\quad For each $u \in \alphabet$, let $\sigma_u$ be a strategy witnessing $\widthrel{X}{hu}$.
Let $\tilde{\sigma}$ be a strategy which on histories prefixed by $hu$ follows $\sigma_u$ (otherwise picking arbitrarily).
Then
\begin{align*}
\widthrel{X}{h} &\ge \min_\tau | \{ \pi \st h \pi \in X \land P_{\tilde{\sigma},\tau}(\pi | h) > 0 \} | \\
&= \min_\tau \sum_{u \in \alphabet} | \{ \pi' \st h u \pi' \in X \land \tau(h,u) \cdot P_{\tilde{\sigma},\tau}(\pi' | h u) > 0 \} | \\
&\ge \min_\tau \min_{u \in \alphabet} | \{ \pi' \st h u \pi' \in X \land P_{\tilde{\sigma},\tau}(\pi' | h u) > 0 \} | \\
&= \min_{u \in \alphabet} \min_\tau | \{ \pi' \st h u \pi' \in X \land P_{\sigma_u,\tau}(\pi' | h u) > 0 \} | \\
&= \min_{u \in \alphabet} \widthrel{X}{hu}.
\end{align*}

For the other direction, let $\tilde{\sigma}$ be a strategy witnessing $\widthrel{X}{h}$, and define $\tilde{u} = \arg\min_{u \in \alphabet} \widthrel{X}{hu}$, and $\tau_{\tilde{u}} = \arg\min_\tau | \{ \pi' \st h \tilde{u} \pi' \in X \land P_{\tilde{\sigma},\tau}(\pi' | h \tilde{u}) > 0 \} |$.
Let $\tilde{\tau}$ be a strategy which on history $h$ picks $\tilde{u}$ and on histories prefixed with $h \tilde{u}$ follows $\tau_{\tilde{u}}$ (otherwise picking arbitrarily).
Then
\begin{align*}
\widthrel{X}{h} &\le | \{ \pi \st h \pi \in X \land P_{\tilde{\sigma},\tilde{\tau}}(\pi | h) > 0 \} | \\
&= \sum_{u \in \alphabet} | \{ \pi' \st h u \pi' \in X \land \tau(h, u) \cdot P_{\tilde{\sigma},\tilde{\tau}}(\pi' | h u) > 0 \} | \\
&= | \{ \pi' \st h \tilde{u} \pi' \in X \land P_{\tilde{\sigma},\tilde{\tau}}(\pi' | h \tilde{u}) > 0 \} | \\
&= | \{ \pi' \st h \tilde{u} \pi' \in X \land P_{\tilde{\sigma},\tau_{\tilde{u}}}(\pi' | h \tilde{u}) > 0 \} | \\
&\le \widthrel{X}{h \tilde{u}} \\
&= \min_{u \in \alphabet} \widthrel{X}{h u}.
\end{align*}
\end{enumerate}
\qed
\end{proof}

\lemmaPartition*
\begin{proof}
Index the elements of $\alphabet$ via some canonical order as $(u_j)_{0 \le j < \ell}$ for some $\ell \ge 1$.
We first construct the partition $\sum_{j < \ell} m^\valids_j$ of $m^\valids$.
Find the greatest $k \le \ell$ such that $\sum_{j < k} \widthrel{\valids}{h u_j} \le m^\valids$.
This is well-defined, since if $k = 0$ then the sum is zero and the condition is satisfied.
If $\sum_{j < k} \widthrel{\valids}{h u_j} = m^\valids$ we put $m^\valids_j = \widthrel{\valids}{h u_j}$ for $j < k$ and $m^\valids_j = 0$ for $j \ge k$.
If instead $\sum_{j < k} \widthrel{\valids}{h u_j} < m^\valids$ we must have $k < \ell$, since $\sum_{j < \ell} \widthrel{\valids}{h u_j} = \sum_{u \in \alphabet} \widthrel{\valids}{h u} = \widthrel{\valids}{h} \ge m^\valids$.
Then by the definition of $k$ we have $\sum_{j \le k} \widthrel{\valids}{h u_j} > m^\valids$, so $\widthrel{\valids}{h u_k} > m^\valids - \sum_{j < k} \widthrel{\valids}{h u_j}$.
Therefore we put $m^\valids_j = \widthrel{\valids}{h u_j}$ for $j < k$, $m^\valids_k = m^\valids - \sum_{j < k} \widthrel{\valids}{h u_j}$, and $m^\valids_j = 0$ for $j > k$.

Now we construct the partition $\sum_{j < \ell} m^\improvs_j$ of $m^\improvs$.
We do this by partitioning the difference $m^\improvs - m^\valids$ along the same lines as above, then adding back $m^\valids_j$ to ensure $m^\improvs_j \ge m^\valids_j$.
Let $d_j = \widthrel{\improvs}{h u_j} - m^\valids_j$.
Since $m^\valids_j \le \widthrel{\valids}{h u_j} \le \widthrel{\improvs}{h u_j}$, we have $d_j \ge 0$.
Find the greatest $k \le \ell$ such that $\sum_{j < k} d_j \le m^\improvs - m^\valids$.
This is well-defined since if $k = 0$ the sum is zero, and $m^\improvs - m^\valids \ge 0$ by assumption.
If $\sum_{j < k} d_j = m^\improvs - m^\valids$ we put $m^\improvs_j = m^\valids_j + d_j$ for $j < k$ and $m^\improvs_j = m^\valids_j$ for $j \ge k$.
This clearly satisfies $m^\valids_j \le m^\improvs_j \le \widthrel{\improvs}{h u_j}$, and $\sum_{j < \ell} m^\improvs_j = \sum_{j < k} (m^\valids_j + d_j) + \sum_{j \ge k} m^\valids_j = \sum_{j < \ell} m^\valids_j + \sum_{j < k} d_j = m^\valids + (m^\improvs - m^\valids) = m^\improvs$ as desired.
If instead $\sum_{j < k} d_j < m^\improvs - m^\valids$ we must have $k < \ell$, since $\sum_{j < \ell} d_j = \sum_{j < \ell} (\widthrel{\improvs}{h u_j} - m^\valids_j) = \sum_{u \in \alphabet} \widthrel{\improvs}{h u} - \sum_{j < \ell} m^\valids_j = \widthrel{\improvs}{h} - m^\valids \ge m^\improvs - m^\valids$.
Then by the definition of $k$ we have $\sum_{j \le k} d_j > m^\improvs - m^\valids$, so $d_k > m^\improvs - m^\valids - \sum_{j < k} d_j$.
Therefore we put $m^\improvs_j = m^\valids_j + d_j$ for $j < k$, $m^\improvs_k = m^\valids_k + (m^\improvs - m^\valids - \sum_{j < k} d_j)$, and $m^\improvs_j = m^\valids_j$ for $j > k$.
Again this satisfies $m^\valids_j \le m^\improvs_j \le \widthrel{\improvs}{h u_j}$, and $\sum_{j < \ell} m^\improvs_j = \sum_{j < k} (m^\valids_j + d_j) + (m^\valids_k + (m^\improvs - m^\valids - \sum_{j < k} d_j)) + \sum_{j > k} m^\valids_j = \sum_{j < \ell} m^\valids_j + (m^\improvs - m^\valids) = m^\valids + (m^\improvs - m^\valids) = m^\improvs$ as desired.

These partitions are canonical since the values of $k$ used in each construction are uniquely determined (and the ordering of $\alphabet$ is fixed).
Also, $k$ may be found by a linear search from $0$ up to $\ell$, which has value at most $|\alphabet|$.
The quantities $\widthrel{\improvs}{h u_j}$ all have polynomial bitwidth (they are bounded above by $|\alphabet|^n$), so the arithmetic above can be done in polynomial time.
Therefore the total time needed to construct the partitions is polynomial relative to oracles for $\widthrel{\improvs}{\cdot}$ and $\widthrel{\valids}{\cdot}$. \qed
\end{proof}

\lemmaSigmaHard*
\begin{proof}
First we show by induction on $i$ that for all plays $h \pi \in \possible_{\hat{\sigma},\tau}$ with $|h| = i$, we have:
\begin{itemize}
\item $0 \le m^\valids(h) \le m^\improvs(h) \le \widthrel{\improvs}{h}$ and $m^\valids(h) \le \widthrel{\valids}{h}$ (so that the partitions used to define $m^\valids(hu)$ and $m^\improvs(hu)$ exist);
\item $t(h), m^\improvs(h) > 0$.
\end{itemize}
In the base case $i = 0$, we must have $h = \emptyword$.
Then $m^\valids(\emptyword) = \width{\valids} \ge 0$ and $m^\improvs(\emptyword) = \width{\improvs} \ge 1 / \rho > 0$.
To show $t(\emptyword) = 1$, there are three cases.
If $\width{\valids} = 0$, then $\alpha = 0$ and $\width{\improvs} - \width{\valids} = \width{\improvs} \ge 1 / \rho > 0$, so $\beta = 1 / \width{\improvs}$ and $t(\emptyword) = \beta \width{\improvs} = 1$.
If $\width{\improvs} - \width{\valids} = 0$, then $\beta = 0$ and $\width{\valids} = \width{\improvs} \ge 1 / \rho$, so $\alpha = 1 / \width{\valids}$ and $t(\emptyword) = \alpha \width{\valids} = 1$.
Otherwise $\alpha = \min(\rho, 1 / \width{\valids})$ and $\beta = (1 - \alpha \width{\valids}) / (\width{\improvs} - \width{\valids})$, so $t(\emptyword) = \alpha \width{\valids} + (1 - \alpha \width{\valids}) = 1$.
Therefore we always have $t(\emptyword) = 1$.

Now take any play $h \pi \in \possible_{\hat{\sigma},\tau}$ with $|h| = i < n$ and suppose the hypothesis holds.
If it is the adversary's turn after $h$, then if the adversary outputs $u \in \alphabet$ we have $m^\valids(h) = m^\valids(hu)$ and $m^\improvs(h) = m^\improvs(hu)$.
So since $m^\valids(h) \le \widthrel{\valids}{h} = \min_{v \in \alphabet} \widthrel{\valids}{hv} \le \widthrel{\valids}{hu}$ and $m^\improvs(h) \le \widthrel{\improvs}{h} = \min_{v \in \alphabet} \widthrel{\improvs}{hv} \le \widthrel{\improvs}{hu}$, the hypothesis holds in the next step.
If instead it is our turn after $h$ and we output $u \in \alphabet$, then $m^\valids(hu)$ and $m^\improvs(hu)$ are given by Lemma \ref{lemma:partition} and $0 \le m^\valids(hu) \le m^\improvs(hu) \le \widthrel{\improvs}{hu}$ and $m^\valids(hu) \le \widthrel{\valids}{hu}$ by construction.
Furthermore $t(hu) > 0$, since if $t(hu) = 0$ then $\hat{\sigma}$ has probability zero to output $u$, a contradiction.
This implies $m^\improvs(hu) > 0$, since if $m^\improvs(hu) = 0$ then $m^\valids(hu) = 0$ and so $t(hu) = 0$.
Therefore by induction we always have $0 \le m^\valids(h) \le m^\improvs(h) \le \widthrel{\improvs}{h}$, $m^\valids(h) \le \widthrel{\valids}{h}$, and $t(h), m^\improvs(h) > 0$.

Now for any history $h \in \histories$ after which it is our turn, by construction the quantities $m^\valids(hu)$ and $m^\improvs(hu)$ for $u \in \alphabet$ form partitions of $m^\valids(h)$ and $m^\improvs(h)$ respectively.
So $\sum_{u \in \alphabet} t(hu) = \sum_{u \in \alphabet} \alpha m^\valids(hu) + \beta (m^\improvs(hu) - m^\valids(hu)) = \alpha m^\valids(h) + \beta (m^\improvs(h) - m^\valids(h)) = t(h) > 0$.
So $\hat{\sigma}(h, \cdot)$ is a probability distribution over $\alphabet$, and $\hat{\sigma}$ is a well-defined strategy.

Finally, take any play $\pi \in \possible_{\hat{\sigma},\tau}$.
As shown above we have $\widthrel{\improvs}{\pi} \ge m^\improvs(\pi) > 0$, and since $|\pi| = n$ this implies $\pi \in \improvs$.
Therefore $P_{\hat{\sigma},\tau}(\improvs) = 1$. \qed
\end{proof}

\lemmaSigmaSoft*
\begin{proof}
We prove by induction on $i$ in decreasing order that for all plays $h \pi \in \possible_{\hat{\sigma},\tau}$ with $|h| = i$, $\sum_{\rho \st h \rho \in \valids} P_{\hat{\sigma},\tau}(\rho | h) \ge \alpha m^\valids(h) / t(h)$.
In the base case $i = n$, by Lemma \ref{lemma:sigma-hard} we must have $h \in \improvs$, so $m^\valids(h) \le m^\improvs(h) \le \widthrel{\improvs}{h} = 1$.
If $m^\valids(h) = 0$ the hypothesis holds trivially.
Otherwise $m^\valids(h) = 1$, so $t(h) = \alpha$ and since $m^\valids(h) \le \widthrel{\valids}{h}$ we must have $h \in \valids$.
Therefore, letting $\rho = \emptyword$ we have $h \rho \in \valids$, so $P_{\hat{\sigma},\tau}(\rho | h) = 1 = \alpha m^\valids(h) / t(h)$ and the hypothesis again holds.

Now take any play $h \pi \in \possible_{\hat{\sigma},\tau}$ with $|h| = i < n$.
If $h \in \improvs$ then the hypothesis holds as above.
Otherwise if it is the adversary's turn after $h$, then $m^\valids(h) = m^\valids(hu)$, $m^\improvs(h) = m^\improvs(hu)$, and $t(h) = t(hu)$ for any $u \in \alphabet$.
By hypothesis, for any play $h u \pi' \in \possible_{\hat{\sigma},\tau}$ we have $\sum_{\rho \st h u \rho \in \valids} P_{\hat{\sigma},\tau}(\rho | h u) \ge \alpha m^\valids(hu) / t(hu) = \alpha m^\valids(h) / t(h)$.
So
\begin{align*}
\sum_{\rho \st h \rho \in \valids} P_{\hat{\sigma},\tau}(\rho | h) &= \sum_{u \in \alphabet} \; \sum_{\rho' \st h u \rho' \in \valids} \tau(h, u) \cdot P_{\hat{\sigma},\tau}(\rho' | h u) \\
&= \sum_{u \in \alphabet} \tau(h, u) \sum_{\rho' \st h u \rho' \in \valids} P_{\hat{\sigma},\tau}(\rho' | h u) \\
&\ge \sum_{u \in \alphabet} \tau(h, u) \cdot \frac{\alpha m^\valids(h)}{t(h)} \\
&= \frac{\alpha m^\valids(h)}{t(h)} \sum_{u \in \alphabet} \tau(h, u) \\
&= \frac{\alpha m^\valids(h)}{t(h)}
\end{align*}
as desired.
If instead it is our turn after $h$, then if we output $u \in \alphabet$ we update $m^\valids$ to $m^\valids_u$, so $m^\valids(hu) = m^\valids_u(h)$.
Then by hypothesis we have $\sum_{\rho \st h u \rho \in \valids} P_{\hat{\sigma},\tau}(\rho | h u) \ge \alpha m^\valids(hu) / t(hu) = \alpha m^\valids_u(h) / t(hu)$.
So
\begin{align*}
\sum_{\rho \st h \rho \in \valids} P_{\hat{\sigma},\tau}(\rho | h) &= \sum_{u \in \alphabet} \; \sum_{\rho' \st h u \rho' \in \valids} \hat{\sigma}(h, u) \cdot P_{\hat{\sigma},\tau}(\rho' | h u) \\
&= \sum_{u \in \alphabet} \hat{\sigma}(h, u) \sum_{\rho' \st h u \rho' \in \valids} P_{\hat{\sigma},\tau}(\rho' | h u) \\
&\ge \sum_{u \in \alphabet} \hat{\sigma}(h, u) \cdot \frac{\alpha m^\valids_u(h)}{t(hu)} \\
&= \alpha \sum_{u \in \alphabet} \frac{t(hu)}{t(h)} \cdot \frac{m^\valids_u(h)}{t(hu)} \\
&= \frac{\alpha}{t(h)} \sum_{u \in \alphabet} m^\valids_u(h) \\
&= \frac{\alpha m^\valids(h)}{t(h)}
\end{align*}
again as desired.
Therefore by induction this holds for every $i$, and in particular for $i = 0$.

Since every play $\pi \in \possible_{\hat{\sigma},\tau}$ is of the form $\emptyword \pi$, noting that $t(\emptyword) = 1$ as shown in Lemma \ref{lemma:sigma-hard} we have $P_{\hat{\sigma},\tau}(\valids) = \sum_{\pi \st \emptyword \pi \in \valids} P_{\hat{\sigma},\tau}(\pi | \emptyword) \ge \alpha m^\valids(\emptyword) / t(\emptyword) = \alpha \width{\valids} = \min(\rho \width{\valids}, 1)$. \qed
\end{proof}

\lemmaSigmaRandom*
\begin{proof}
We prove by induction on $i$ in decreasing order that for all plays $h \pi \in \possible_{\hat{\sigma},\tau}$ with $|h| = i$, $P_{\hat{\sigma},\tau}(\pi | h) \le \max(\alpha,\beta) / t(h)$.
In the base case $i = n$, by Lemma \ref{lemma:sigma-hard} we must have $h \in \improvs$.
Then $m^\improvs(h) = 1$, since $m^\improvs(h) \le \widthrel{\improvs}{h} = 1$ and $m^\improvs(h) > 0$ if $h$ can be generated by $\hat{\sigma}$ (as shown in Lemma \ref{lemma:sigma-hard}).
So $t(h) = \alpha m^\valids(h) + \beta (1 - m^\valids(h))$, and thus either $t(h) = \alpha$ or $t(h) = \beta$ (depending on whether $m^\valids(h) = 1$ or $m^\valids(h) = 0$).
In either case $\max(\alpha,\beta) / t(h) \ge 1$, so $P_{\hat{\sigma},\tau}(\pi | h) \le \max(\alpha,\beta) / t(h)$ as desired.

Now take any play $h \pi \in \possible_{\hat{\sigma},\tau}$ with $|h| = i < n$.
Since $|h| < n$ the play is of the form $h u \pi'$ for some $u \in \alphabet$, and by hypothesis $P_{\hat{\sigma},\tau}(\pi' | h u) \le \max(\alpha,\beta) / t(hu)$.
Now if it is the adversary's turn after $h$, then $m^\valids(hu) = m^\valids(h)$ and $m^\improvs(hu) = m^\improvs(h)$, so $t(hu) = t(h)$ and therefore $P_{\hat{\sigma},\tau}(\pi | h) = \tau(h, u) \cdot P_{\hat{\sigma},\tau}(\pi' | h u) \le P_{\hat{\sigma},\tau}(\pi' | h u) \le \max(\alpha,\beta) / t(hu) = \max(\alpha,\beta) / t(h)$ as desired.
If instead it is our turn after $h$, then
\begin{align*}
P_{\hat{\sigma},\tau}(\pi | h) &= \hat{\sigma}(h, u) \cdot P_{\hat{\sigma},\tau}(\pi' | h u) \\
&= \frac{t(hu)}{t(h))} \cdot P_{\hat{\sigma},\tau}(\pi' | h u) \\
&\le \frac{t(hu)}{t(h)} \cdot \frac{\max(\alpha,\beta)}{t(hu)} \\
&= \frac{\max(\alpha,\beta)}{t(h)}
\end{align*}
again as desired.
So by induction the hypothesis holds for every $i \in \{0, \dots, n\}$, and in particular for $i = 0$.

Since every play $\pi \in \possible_{\hat{\sigma},\tau}$ is of the form $\emptyword \pi$, we have $P_{\hat{\sigma},\tau}(\pi) = P_{\hat{\sigma},\tau}(\pi | \emptyword) \le \max(\alpha,\beta) / t(\emptyword) = \max(\alpha,\beta)$ (as $t(\emptyword) = 1$).
Recall that $\alpha = \min(\rho, 1 / \width{\valids}) \le \rho$ and $\beta = (1 - \alpha \width{\valids}) / (\width{\improvs} - \width{\valids})$, with the convention that $\alpha = 0$ if $\width{\valids} = 0$ and $\beta = 0$ if $\width{\improvs} - \width{\valids} = 0$.
If $\alpha = \rho$, then $\beta = (1 - \rho \width{\valids}) / (\width{\improvs} - \width{\valids}) \le (1 - \rho \width{\valids}) / ((1/\rho) - \width{\valids}) = \rho$.
If instead $\alpha = 1 / \width{\valids}$, then $\beta = 0$.
Finally, if $\alpha = 0$ then $\width{\valids} = 0$ and $\beta = 1 / \width{\improvs} \le \rho$.
So $\max(\alpha,\beta) \le \rho$, and therefore $P_{\hat{\sigma},\tau}(\pi) \le \rho$ for every $\pi \in \possible_{\hat{\sigma},\tau}$.
In fact this holds for all plays $\pi \in \alphabet^n$, since if $\pi \not \in \possible_{\hat{\sigma},\tau}$ then $P_{\hat{\sigma},\tau}(\pi) = 0$. \qed
\end{proof}

\theoremPspaceScheme*
\begin{proof}
We implement the operations required by Theorem \ref{theorem:generic-scheme}.
Intersection is simple: we run the algorithms $\mathcal{X}$ and $\mathcal{Y}$ and accept iff both do.
The resulting algorithm runs in polynomial space, since $\mathcal{X}$ and $\mathcal{Y}$ do, and can be constructed in polynomial space (indeed, time).

For width measurement, we compute $\widthrel{X}{h}$ using an arithmetization of the usual $\PSPACE$ algorithm for $\QBF$, replacing $\lor$ by $+$ at $\exists$ nodes in the recursive tree and $\land$ by $\min$ at $\forall$ nodes.
Specifically, if it is our turn after $h$ (corresponding to an $\exists$ quantifier in a $\QBF$), we recursively compute $\widthrel{X}{hu}$ for each $u \in \alphabet$ and return $\sum_{u \in \alphabet} \widthrel{X}{hu} = \widthrel{X}{h}$.
If instead it is the adversary's turn, we again recursively compute $\widthrel{X}{hu}$ for each $u \in \alphabet$ but now return $\min_{u \in \alphabet} \widthrel{X}{hu} = \widthrel{X}{h}$.
Finally, in the base case $|h| = n$ we have $\widthrel{X}{h} = \ind{h \in X}$ and so simply invoke $\mathcal{X}$ to determine in polynomial space whether $h \in X = L(\mathcal{X}) \cap \alphabet^n$.
As in the $\QBF$ algorithm, the recursive tree has polynomial depth, and since $\widthrel{X}{h} \le |\alphabet|^n$ we need only polynomial space to remember partial results along the current path through the tree.
So we can compute $\widthrel{X}{h}$ in polynomial space. \qed
\end{proof}

Note that for temporal logics, the alphabet $\alphabet$ is usually the set $2^P$ of all possible assignments to finitely-many Boolean propositions $P$, some of which are controlled by the adversary.
Translating specifications in this form to our convention (players alternately picking symbols from a single alphabet, which is more convenient for control improvisation) is straightforward.

\theoremLTLHardness*
\begin{proof}
We show hardness by reduction from $\QBF$.
Given a $\QBF$ $\phi$ with variables numbered $1, \dots, n$, without loss of generality we may assume the quantifiers strictly alternate starting with $\exists$ and that the matrix is in conjunctive normal form, consisting of clauses $c_1, \dots, c_m$.
We can view an assignment to the variables as a length-$n$ trace over a single proposition $p$ indicating the truth of the variable corresponding to the position.
Then for each clause $c_i$ we can construct an $\LTL$ formula $\psi_c$ whose models of length $n$ are exactly the assignments satisfying the clause.
If $c_i$ contains the variables $V^+$ positively and $V^-$ negatively, then we put
\[
\psi_c = \bigvee_{v \in V^+} X^{v-1} p \lor \bigvee_{v \in V^-} X^{v-1} (\lnot p) .
\]
Then putting $\psi = \bigwedge_i \psi_{c_i}$, the length-$n$ models of $\psi$ are exactly the assignments satisfying the matrix of $\phi$.
So we have a winning strategy to generate a play satisfying $\psi$ if and only if $\phi$ is true.
Finally, this construction clearly can be done in polynomial time. \qed
\end{proof}

\theoremNFA*
\begin{proof}
We give a reduction as in Theorem \ref{theorem:ltl-hardness} from a $\QBF$ $\phi$ in existential prenex normal form, but with the matrix in disjunctive normal form.
Using the method of \cite{sharpNFA}, we can construct in polynomial time an $\NFA$ $N$ which accepts exactly the satisfying assignments of the matrix of $\phi$.
Then we have a winning strategy to generate a play in $L(N)$ if and only if $\phi$ is true. \qed
\end{proof}

\end{document}